\documentclass[transmag,10pt,a4paper,final]{IEEEtran}

\usepackage{silence}
\usepackage{cite}

\usepackage{amsmath,amssymb,amsthm,fixmath}

\usepackage{color,colortbl}
\usepackage{xcolor}

\usepackage[inline]{enumitem}

\usepackage{siunitx}
\DeclareSIUnit{\dBm}{dBm}

\usepackage{graphicx}
\usepackage[labelformat=simple]{subcaption}

\usepackage{epstopdf}
\usepackage{cuted}
\usepackage{multirow,booktabs}
\usepackage{diagbox}
\usepackage{makecell}
\usepackage{hhline}
\usepackage{array} 
\newcolumntype{x}{!{\vrule width 2px}}
\newcolumntype{y}{!{\vrule width 1.5px}}

\usepackage{optidef}

\usepackage[hidelinks]{hyperref}

\usepackage[shortcuts,acronym,automake]{glossaries}
\makeglossaries
\newacronym{ai}{AI}{artificial intelligence}
\newacronym{awgn}{AWGN}{additive white Gaussian noise}
\newacronym{bcd}{BCD}{block coordinate descent}
\newacronym{bs}{BS}{base station}
\newacronym{cp}{CP}{control plane}
\newacronym{crc}{CRC}{cyclic redundancy check}
\newacronym{csi}{CSI}{channel state information}
\newacronym{csit}{CSIT}{channel state information at transmitter}
\newacronym{dft-s-ofdm}{DFT-s-OFDM}{Discrete Fourier Transform-spread-OFDM}
\newacronym{fbl}{FBL}{finite blocklength}
\newacronym{gan}{GAN}{generative adversarial network}
\newacronym{ibl}{IBL}{infinite blocklength}
\newacronym{lfp}{LFP}{leakage-failure probability}
\newacronym{lut}{LUT}{look-up table}
\newacronym{mac}{MAC}{Medium Accesss Control}
\newacronym{mimo}{MIMO}{multi-input multi-output}
\newacronym{ml}{ML}{machine learning}
\newacronym{mm}{MM}{Minorize-Maximization}
\newacronym{noma}{NOMA}{non-orthogonal multi-access}
\newacronym{nom}{NOM}{non-orthogonal multiplexing}
\newacronym{ofdm}{OFDM}{orthogonal frequency-division multiplexing}
\newacronym{ofdma}{OFDMA}{orthogonal frequency-division multiple access}
\newacronym{oma}{OMA}{orthogonal multiple access}
\newacronym{papr}{PAPR}{Peak-to-Average Power Ratio}
\newacronym{per}{PER}{packet error rate}
\newacronym{phy}{PHY}{physical}
\newacronym{pld}{PLD}{physical layer deception}
\newacronym{pls}{PLS}{physical layer security}
\newacronym{prb}{PRB}{physical resource block}
\newacronym{sic}{SIC}{successive interference cancellation}
\newacronym{sinr}{SINR}{signal-to-interference-and-noise ratio}
\newacronym{snr}{SNR}{signal-to-noise ratio}
\newacronym{tdma}{TDMA}{time-division multiple access}
\newacronym{ue}{UE}{user equipment}
\newacronym{up}{UP}{user plane}
\newacronym{urllc}{URLLC}{ultra-reliable low-latency communication}

\usepackage{tcolorbox}
\tcbuselibrary{many}
\newtheorem{theorem}{Theorem}
\newtheorem{lemma}{Lemma}

\usepackage[norelsize,linesnumbered,ruled]{algorithm2e}
\SetKwRepeat{Do}{do}{while}
\makeatletter
\newcommand{\removelatexerror} {\let\@latex@error\@gobble}
\makeatother

\newcommand{\superscript}[1]{^{\mathrm{#1}}}
\newcommand{\subscript}[1]{_{\mathrm{#1}}}

\usepackage[normalem]{ulem}
\newcommand{\revise}[2]{{\color{red}\sout{#1}}{\color{blue}#2}} 
\renewcommand{\revise}[2]{{\color{blue}#2}} 
\renewcommand{\revise}[2]{#2} 
\newcommand{\reviseRT}[2]{{\color{blue}#2}} %
\renewcommand{\reviseRT}[2]{#2} %

\usepackage[textsize=tiny,colorinlistoftodos]{todonotes}
\makeatletter
\define@key{todonotes}{bh}[]{
	\setkeys{todonotes}{author=\textbf{Bin}, color=lime!30}}%
\define@key{todonotes}{yz}[]{
	\setkeys{todonotes}{author=\textbf{Yao}, color=blue!30}}%
\define@key{todonotes}{wc}[]{
	\setkeys{todonotes}{author=\textbf{Wenwen}, color=green!30}}%
\makeatother

\newif\ifreviewmode
\reviewmodetrue 

\ifreviewmode
\else
  \renewcommand{\todo}[1]{} 
  \renewcommand{\revise}[2]{#2} 
\fi

\newcommand\bob{\subscript{Bob}}
\newcommand\eve{\subscript{Eve}}
\newcommand\lf{\subscript{LF}}

\begin{document}

\title{Physical Layer Deception with Non-Orthogonal Multiplexing}

\author{
	 \author{Wenwen~Chen,
		Bin~Han, 
		Yao~Zhu,
		Anke~Schmeink, 
        Giuseppe~Caire, 
		and
		Hans~D.~Schotten
		\thanks{W. Chen, B. Han, and H. D. Schotten are with University of Kaiserslautern (RPTU), Germany. Y. Zhu and A. Schmeink are with RWTH Aachen University, Germany. G. Caire is with Technical University of Berlin, Germany. H. D. Schotten is with the German Research Center for Artificial Intelligence (DFKI), Germany. 
		B. Han (bin.han@rptu.de) and Y. Zhu (yao.zhu@inda.rwth-aachen.de) are the corresponding authors.
		}
	}
	
	\IEEEauthorblockN{
        Wenwen~Chen\IEEEauthorrefmark{1},
		Bin~Han\IEEEauthorrefmark{1},
		Yao~Zhu\IEEEauthorrefmark{2},
		Anke~Schmeink\IEEEauthorrefmark{2},	
Giuseppe~Caire\IEEEauthorrefmark{4}
and~Hans~D.~Schotten\IEEEauthorrefmark{1}\IEEEauthorrefmark{3}
	}
	
	\IEEEauthorblockA{
		\IEEEauthorrefmark{1}RPTU Kaiserslautern-Landau, \IEEEauthorrefmark{2}RWTH Aachen University, \IEEEauthorrefmark{3}German Research Center of Artificial Intelligence (DFKI),
        \IEEEauthorrefmark{4}Technical University of Berlin
	}
}

\bstctlcite{IEEEexample:BSTcontrol}

\maketitle

\begin{abstract}
\Ac{pls} is a promising technology to secure wireless communications by exploiting the physical properties of the wireless channel. However, the passive nature of \ac{pls} creates a significant imbalance between the effort required by eavesdroppers and legitimate users to secure data. To address this imbalance, in this article, we propose a novel framework of \ac{pld}, which combines \ac{pls} with deception technologies to actively counteract wiretapping attempts. Combining a two-stage encoder with randomized ciphering and non-orthogonal multiplexing, the \ac{pld} approach enables the wireless communication system to proactively counter eavesdroppers with deceptive messages. Relying solely on the superiority of the legitimate channel over the eavesdropping channel, the \ac{pld} framework can effectively protect the confidentiality of the transmitted messages, even against eavesdroppers who possess knowledge equivalent to that of the legitimate receiver. We prove the validity of the \ac{pld} framework with in-depth analyses and demonstrate its superiority over conventional \ac{pls} approaches with comprehensive numerical benchmarks.
\end{abstract}


\IEEEpeerreviewmaketitle

\glsresetall

\section{Introduction}\label{sec:introduction}

\Ac{pls} is gaining prominence within wireless communication, aiming to secure transmissions by leveraging physical channel properties, thus offering a \revise{fresh layer}{new paradigm} of security without relying on traditional cryptographic methods. This technological trend is increasingly pivotal in modern wireless networks~\cite{HFA2019classification}.
However, future wireless networks are expected to support ultra-reliable and low-latency communications (URLLC)~\cite{SO2021urllc_intro}, where the transmitted packets usually consist of mission-critical information with small amount of bits, e.g., command signals for the actuator or real-time measurement from sensors. Due to the stringent delay requirement, only a limited number of blocklength is assigned to these packets, which indicates the classic assumption of infinite  blocklength no longer holds~\cite{PPV2010channel}. The short-packet transmissions operated in the so-called \ac{fbl} regime may still suffer from potential decoding error or potential leakage, even if \revise{Bob's channel is stronger than Eve's channel}{the legitimate user (\emph{Bob}) has a stronger channel than the eavesdropper (\emph{Eve})}. To characterize the \ac{pls} performance with \ac{fbl} codes, the authors in~\cite{YSP2019wiretap} provide the  bounds of achievable security rate under a given leakage probability and a given error probability, which offers a more general expression in the \ac{fbl} regime than the Wyner's secrecy capacity.


On the other hand, driven by the inherent interference along with the superimposed transmission,   
the non-orthogonal multiplexing technologies shows the potential to be the new opportunities for enhancing the \ac{pls} performance~\cite{LDQ2024PLSNOMA}. Leveraging the non-orthogonality of the signals, 
the  communication secrecy can be enhanced by constructively engineering interference in the wiretap channel. 
In fact, it has been proven from an information-theoretic perspective that transmitting open \revise{}{(public)} messages simultaneously with confidential messages can improve overall secrecy performance when using a security-oriented precoder design~\cite{XYW2022PLSpercoder}.

Despite the latest advances of \ac{pls} offering enhanced passive security, a notable imbalance still remains, as eavesdroppers \revise{}{(especially passive ones)} can attempt to wiretap with barely any risk of exposure\revise{}{~\cite{CWS+2018ghostbuster}}, and significantly lower effort compared to the extensive measures \revise{}{and costs~\cite{WBZ+2019survey,ZGA+2011throughput}} taken by the network and legitimate users to secure data. This imbalance necessitates a strategic pivot towards integrating active defense mechanisms, like deception technologies, into the wireless security framework. Deception technologies are designed to mislead and distract potential eavesdroppers by fabricating data or environments, thus protecting genuine data. These technologies can also entice eavesdroppers into revealing themselves, offering a proactive approach to maintaining security integrity~\cite{WL2018cyber}.

\revise{
	In 2023, we proposed a novel framework of \ac{pld}~\cite{HZS+2023nonorthogonal}, which pioneered to combine \ac{pls} and deception technologies. The \ac{pld} framework enables to deceive eavesdroppers and actively counteract their wiretapping attempts. Leveraging our insights of \ac{pls} in the \ac{fbl} regime, a joint optimization of the encryption coding rate and the power allocation was formulated and solved in~\cite{HZS+2023nonorthogonal}, to simultaneously achieve high secured reliability and effective deception. Nevertheless, as a preliminary effort, this work is still limited in several aspects. First, it lacks in-depth discussion about the error model of wiretap channels in the presence of deceptive ciphering. Second, the optimization problem is set up with a simple objective function that straightforwardly combines the secrecy performance and the deception performance, which lacks flexibility and adaptability to various practical scenarios. Third, the \ac{pld} framework is only discussed on a simple point-to-point communication scenario, without discussion about practical implementation and deployment in multi-access networks.
	}{
	While our previous work~\cite{HZS+2023nonorthogonal} pioneered the integration of \ac{pls} and deception technologies through the \ac{pld} framework, it left several critical aspects unexplored:
	\begin{enumerate}
		\item A comprehensive error model for wiretap channels incorporating deceptive ciphering was not developed.
		\item The optimization problem lacked flexibility in balancing secrecy and deception performance due to its simplistic objective function.
		\item It lacks of discussion about technical challenges of practical implementation and deployment, e.g., regarding the ciphering codec design, imperfect \ac{csi}, and multi-access network scenarios.
	\end{enumerate}
}

\revise{
	In this article, we extend our previous work and address the aforementioned limitations. With respect to \cite{HZS+2023nonorthogonal}, the main novel contributions of this article are summarized as follows:
	\begin{enumerate}
		\item We detail the system model with the distinguished transmission schemes with both activated and deactivated deceptive ciphering, and provide a comprehensive reception error model of the proposed approach in different scenarios.
		\item Instead of simply combining the reception rate and deception rate of both the legitimate user and the eavesdropper into a single objective function, we propose in this article to maximize the effective deception rate, while setting a constraint on the \ac{lfp}. This new setup allows to flexibly adjust the trade-off between the secrecy performance and the deception performance, and to adapt to various practical scenarios. The analyses to the feasibility and convexity of the optimization problem are correspondingly updated, and the numerical solver design also adapted therewith.
		\item Two different solutions of conventional \ac{pls}, instead of only one in \cite{HZS+2023nonorthogonal}, are provided as benchmarks to evaluate the performance of the \ac{pld} framework.
		\item We extend the discussion about the \ac{pld} framework to various aspects of implementation and deployment in practical use scenarios. 
	\end{enumerate}
	}{
		This article extends our previous work to address these limitations. The main novel contributions are as follows:
		\begin{enumerate}
			\item We present a detailed system model and a comprehensive reception error model for scenarios with both activated and deactivated deceptive ciphering.
			\item We propose a novel optimization approach that maximizes the effective deception rate while constraining the leakage-failure probability. This formulation enables flexible adaptation to various practical scenarios and allows for fine-tuning the trade-off between secrecy and deception performance.
			\item We provide extended discussions about the \ac{pld} framework's implementation and deployment considerations in practical scenarios.
		\end{enumerate}		
		Furthermore, we enhance the evaluation of our approach by benchmarking it against two distinct conventional \ac{pls} solutions, thereby providing a more rigorous assessment of the \ac{pld} framework's performance.
}

The remaining contents of this article are organized as follows. We begin with a brief review to related literature in Sec.~\ref{sec:related}, then setup the models and optimization problem in Sec.~\ref{sec:problem}. Afterwards, we present our analyses to the problem and our approach in Sec.~\ref{sec:approach}, which are later numerically validated and evaluated in Sec.~\ref{sec:evaluation}. Finally, we extend our discussion in Sec.~\ref{sec:discussion} regarding various perspectives of practical implementation and deployment of the \ac{pld} paradigm, before concluding this article with Sec.~\ref{sec:conclusion}.

\section{Related Works}\label{sec:related}
The concept of \ac{pls} originates from the seminal work of \emph{Wyner}~\cite{Wyner1975wiretap}, which generalizes \emph{Shannon}'s concept of perfect secrecy~\cite{Shannon1949communication} into a measurable strong secrecy over wiretap channels. Since then, the secrecy performance of communication systems has been studied over different variants of wiretap channels, including binary symmetric channels~\cite{CH1977note}, degraded \ac{awgn} channels~\cite{CH1978gaussian}, fading channels~\cite{GLE2008secrecy}, multi-antenna channels~\cite{ZGA2010secure}, broadcast channels~\cite{CK1978broadcast}, multi-access channels~\cite{TY2008gaussian}, interference channels~\cite{LSP+2009capacity}, relay channels~\cite{Oohama2007capacity}, etc. Thanks to the rich insights gained therefrom, optimization methods have been developed in various perspectives, such as radio resource allocation, beamforming and precoding, antenna/node selection and cooperation, channel coding, etc.~\cite{MFH+2014principles,WBZ+2019survey}, to enhance the secrecy performance of communication systems.

To investigate the transmission performance in \ac{fbl} regime, authors in the landmark work~\cite{PPV2010channel} derive a tight bound with the closed-form expression for the decoding error probability. This expression and its first-order approximation are widely adopted to investigate \ac{fbl} performance, especially with \ac{urllc} applications~\cite{LZL2023eMBBURLLC,LZW2023URLLC+uplink,LLY2023URLLC+OTFS}.  
Following this effort, authors in~\cite{YSP2019FBLPLS} derive the achievable security rate and its tight bounds for both discrete memoryless and Gaussian wiretap channels.  Based on that, abounding works have been done to enhance the \ac{pls} performance in \ac{fbl} regime. 
For example, the authors in~\cite{WLZ2022PLS+covertness} investigate the covertness by keeping the confidential signal below a certain signal-to-noise ratio threshold for the wiretap channel so that Eve can not detect the transmission. In~\cite{OPC2023PLS+reliable}, the interplay between reliability and security is studied, where the joint secure-reliability performance is enhanced by allocating the resource of the transmissions. This interplay is further investigated in~\cite{ZYH+2023trade}, where the concept of trading reliability for security is proposed to characterize the trade-off between security and reliability in \ac{pls} for the short-packet transmissions. 

Another emerging cluster of research focuses on the application of \ac{noma} in \ac{pls}. \ac{noma} is a promising technology that allows multiple users to share the same frequency and time resources, which can significantly increase spectral efficiency. Especially for \ac{pls}, the interference caused by the superposition signals could be beneficial to improve the security~\cite{CWD+2021improving,XYP+2019physical}. Therefore, \ac{noma}-based \ac{pls} has been shown to provide enhanced security compared to conventional approaches. Nevertheless, such studies are also generally considering long codes, leaving \ac{noma}-\ac{pls} in the \ac{fbl} regime a virgin land of research.



In the field of information security, the principles of deception were firstly introduced and well demonstrated by the infamous practices of social engineering by \emph{Mitnick}~\cite{MS2003art}. Later, this concept was transferred by \emph{Cheswick}~\cite{Cheswick1992evening} and \emph{Stoll}~\cite{Stoll2005cuckoo} into defensive applications, which were originally called \emph{honeypots} and thereafter generalized to a broader spectrum of \emph{deception technologies}~\cite{FAL+2018demystifying}. \revise{}{The core principle of deception technologies is misleading and distracting potential attackers with fake targets, e.g. fabricated data with similar features like the confidential data, and therewith protecting genuine information. These technologies can also entice attackers into revealing themselves, offering a proactive approach to maintaining security integrity.} Over the past three decades, deception technologies have been well developed and widely adopted in information systems, across the four layers of network, system, application, and data. Various solutions have been proposed to mitigate, prevent, or to detect cyber attacks. For a comprehensive review on the state-of-the-art of deception technologies, readers are referred to~\cite{FAL+2018demystifying,HKB2018deception,PCZ2019game}. On the physical layer of wireless systems, however, deception technologies are still in their infancy. Besides our preliminary work~\cite{HZS+2023nonorthogonal} that has been introduced in Sec.~\ref{sec:introduction}, the inspiration of exploiting physical characteristics of wireless channels to proactively deceiving eavesdroppers with fake information is also seen in the works of \cite{HFW+2023proactive} and \cite{QMZ+2024adversarial}. More specifically, the former work leverages the spatial diversity of \ac{mimo} systems to attract an eavesdropper to gradually approach a \emph{trap region} where fake messages are received, while the latter designs a \ac{gan} to generate specialized waveform that paralyzes the eavesdropper's recognition model.

\section{Problem Setup}\label{sec:problem}
\subsection{System Model}\label{subsec:sys_model}
We consider a typical wireless eavesdropping scenario. The information source \emph{Alice} sends messages to the sink \emph{Bob} over its legitimate channel $h\bob$ with gain $z\bob=\left\vert h\bob\right\vert^2$. Meanwhile, an attacker \emph{Eve} listens to \emph{Alice} over the eavesdropping channel $h\eve$ with gain $z\eve=\left\vert h\eve\right\vert^2$. Though \emph{Eve} may \revise{}{theoretically} occur at any position, \revise{practically, due to the concern of exposure, \revise{she}{\emph{Eve}} is}{it is in practice} unlikely staying consistently close to \emph{Alice} or \emph{Bob}. \revise{}{To understand this difficulty for \emph{Eve}, we can investigate the downlink scenario for instance, where \emph{Alice} is a fixed \ac{bs} and \emph{Bob} a mobile \ac{ue}. On the one hand, it is a common practice in wireless networks to deploy \emph{secrecy guard zones} in which \emph{Alice} can detect the existence of \emph{Eve} in its vicinity~\cite{ZGA+2011throughput}, or to realize \emph{secrecy protected zones} that either inherently or intentionally the existence of \emph{Eve} around \emph{Alice}~\cite{CCL+2014enhanced}, e.g. by installing the antenna at a physically secured location. On the other hand, \emph{Bob} is generally considered moving randomly across various wireless security models~\cite{BWB+2019modeling}. Due to the anonymity provided by modern wireless networks, it is often challenging for \emph{Eve} to precisely localize \emph{Bob} and to stay close thereto. For the uplink, the roles of \emph{Alice} and \emph{Bob} are exchanged, so the secrecy guard/protected zones apply to \emph{Bob} and the random mobility apply to \emph{Alice}, which also prevents \emph{Eve} from consistently staying close to each of them. Discussions about the links between two fixed nodes (e.g., wireless backhaul) or two mobile devices (e.g., vehicle-to-vehicle communication) are also trivial.} In this context, with proper beamforming, \emph{Alice} is capable of keeping $h\bob$ statistically superior to $h\eve$, which allows to apply \ac{pls} approaches. To enable physical layer deception, \emph{Alice} deploys a two-stage encoder followed by \ac{nom}-based waveforming, as illustrated in Fig.~\ref{fig:alice_model}.

\begin{figure}[!htpb]
	\centering
	\begin{subfigure}[t]{\linewidth}
		\centering
		\frame{\includegraphics[width=\linewidth]{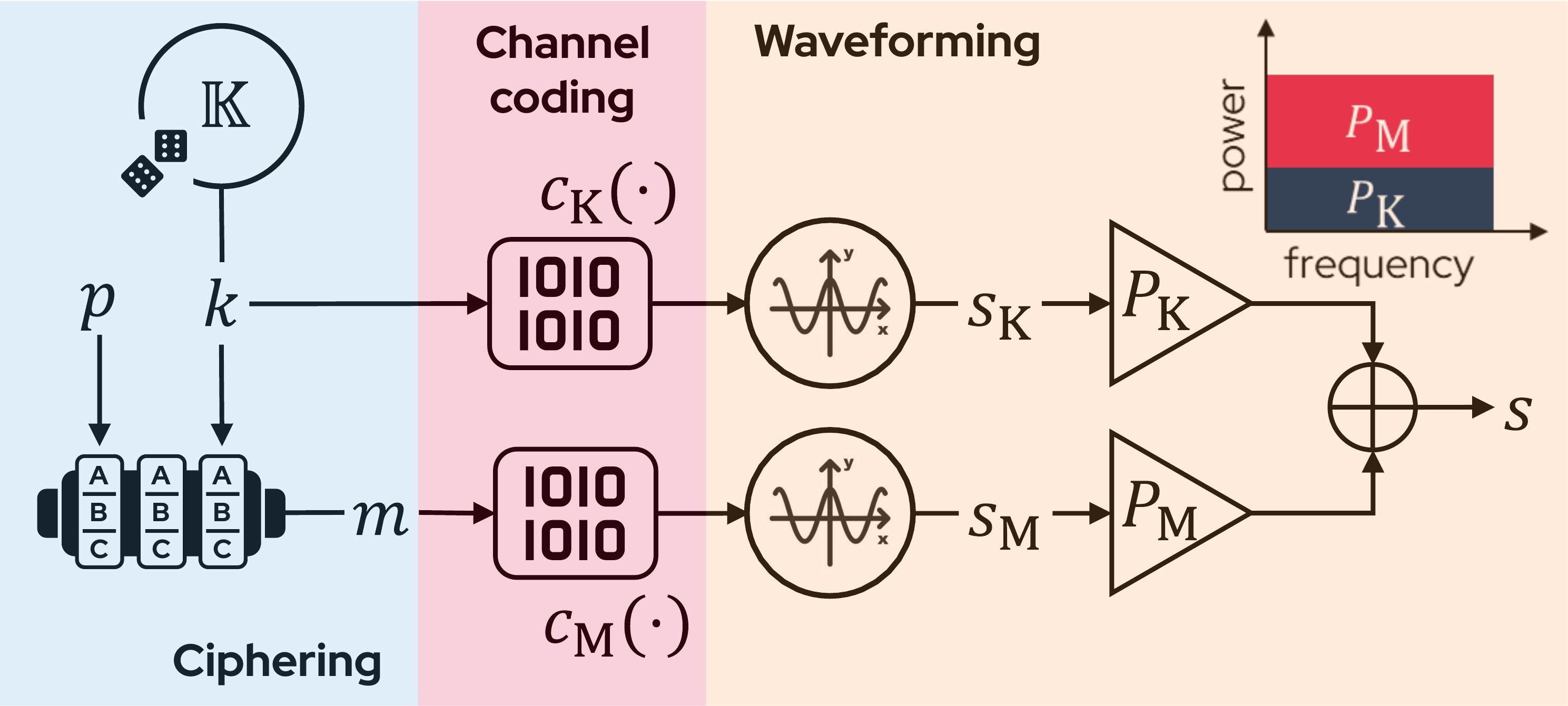}}
		\subcaption{}
		\label{subfig:alice_model_deception_active}
	\end{subfigure}\\
	\begin{subfigure}[t]{\linewidth}
		\centering
		\frame{\includegraphics[width=\linewidth]{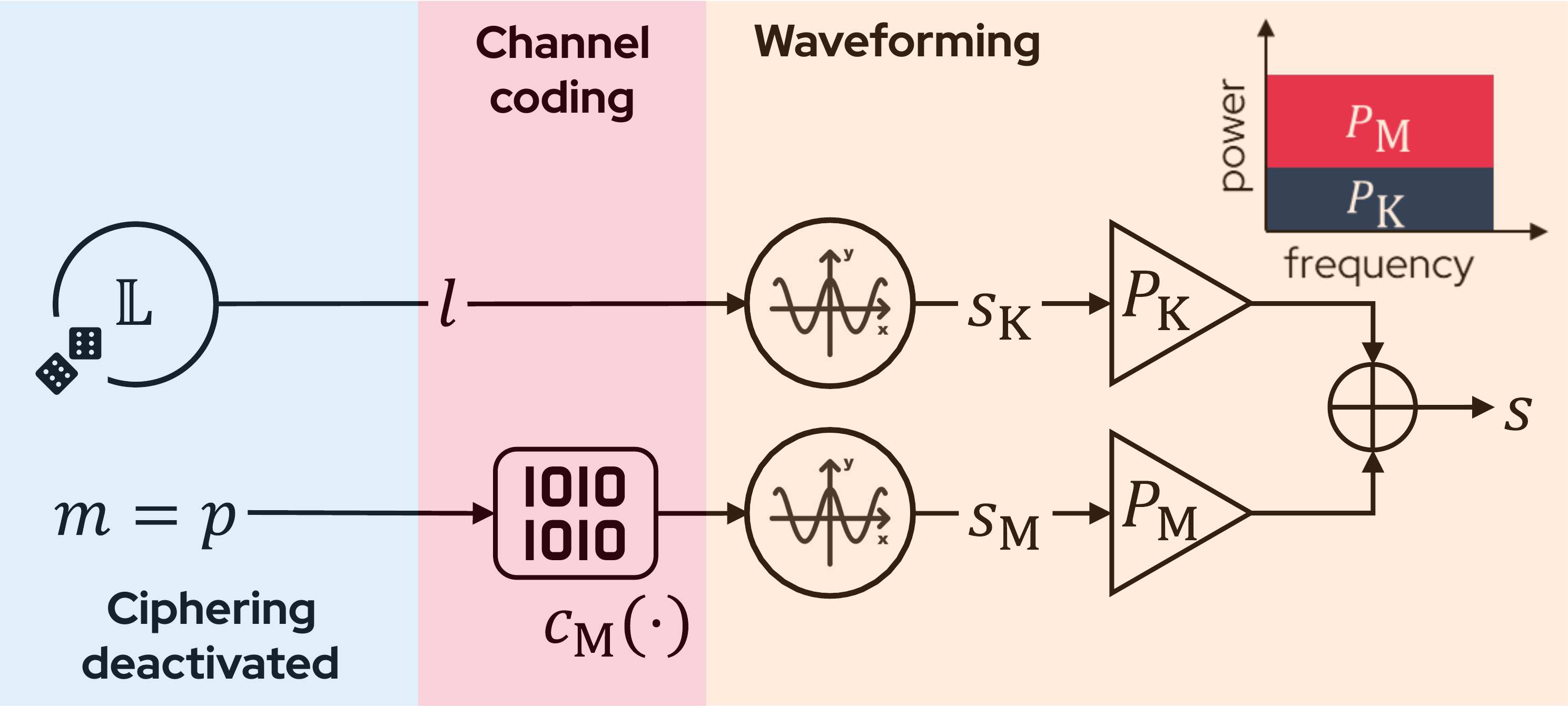}}
		\subcaption{}
		\label{subfig:alice_model_deception_inactive}
	\end{subfigure}
	\caption{The transmitting scheme of \emph{Alice}, with deceptive ciphering \subref{subfig:alice_model_deception_active} activated and \subref{subfig:alice_model_deception_inactive} deactivated, respectively.}
	\label{fig:alice_model}
\end{figure}

The first stage of the encoder is a symmetric key cipherer that can be optionally activated/deactivated. When activated, as illustrated in Fig.~\ref{subfig:alice_model_deception_active}, it encrypts every $d\subscript{P}$-bit plaintext message $p$ with an individually and randomly selected $d\subscript{K}$-bit key $k$ into a $d\subscript{M}$-bit ciphertext 
\begin{equation}\label{eq:ciphering}
	m=f(p,k)\in\mathbb{M},\quad\forall (p,k)\in\left(\mathbb{P}\times\mathbb{K}\right),
\end{equation}
where $\mathbb{M}\subseteq \{0,1\}^{d\subscript{M}}$, $\mathbb{P}\subseteq \{0,1\}^{d\subscript{P}}$, and $\mathbb{K}\subseteq \{0,1\}^{d\subscript{K}}$ are the feasible sets of ciphertext codes, plaintext codes, and keys, respectively.
On the other hand, given the chosen key $k$, the plaintext can be decrypted from the ciphertext through
\begin{equation}
	p=f^{-1}(m,k).
\end{equation}
Especially, the codebooks shall be designed to ensure that
\begin{equation}\label{eq:undetectable_ciphering}
	\mathbb{M}\subseteq\mathbb{P}
\end{equation}
\revise{}{which implies $d\subscript{P}=d\subscript{M}$,} and that $\forall\left(m,k,k'\right)\in\left(\mathbb{M}\times\mathbb{K}^2\right)$ it holds
\begin{equation}\label{eq:unreplacable_key}
	f^{-1}(m,k')\left\vert_{k'\neq k}\neq f^{-1}(m,k).\right.
\end{equation}
The second stage of the encoder is a pair of channel coders, $c\subscript{M}$ and $c\subscript{K}$, which attach error correction redundancies to the ciphertext $m$ and the key $k$, respectively. The two output codewords, both $n$ bits long, are then individually modulated before non-orthogonally multiplexed in the power domain:
\begin{equation}\label{eq:multiplexing}
	s=\frac{P\subscript{M}s\subscript{M}+P\subscript{K}s\subscript{K}}{P\subscript{M}+P\subscript{K}},
\end{equation}
where $s$ is the power-normalized baseband signal to transmit, $s\subscript{M}$ and $s\subscript{K}$ the power-normalized baseband signals carrying the ciphertext and the key, respectively. $P\subscript{M}$, $P\subscript{K}$, and $P\subscript{M}+P\subscript{K}$ are the transmission powers allocated to $s\subscript{M}$, $s\subscript{K}$, and $s$, respectively. Particularly, we set $s\subscript{M}$ as the primary component of the message, and $s\subscript{K}$ the secondary, so that $P\subscript{M}>P\subscript{K}$. On the receiver side, for both $i\in\{\text{Bob}, \text{Eve}\}$:
\begin{equation}
	r_i=s_i*h_i+w_i,
\end{equation}
where $r_i$ and $w_i$ are the baseband signal and equivalent baseband noise received at $i$, respectively.

On the other hand, when the cipherer is deactivated, the transmitter functions as shown in Fig.~\ref{subfig:alice_model_deception_inactive}. The plaintext $p$ is directly inherited as the ciphertext, i.e., $m=p$. 
\reviseRT{Meanwhile, to prevent revealing the status of cipherer through the power profile of $s$, the key component $s\subscript{K}$ is still generated, but not carrying any valid ciphering key $k\in\mathbb{K}$. Instead, a ``litter'' sequence $l\in\mathbb{L}$ is randomly selected to derive $s\subscript{K}$ in this case.}{Remark that though no valid ciphering key is generated in this case, it is risky to simply transmit $s=s\subscript{M}$ without any masking. This is because the power profile of $s$ may reveal the status of cipherer, which can be exploited by \emph{Eve} to infer the plaintext when it estimates the ciphering is deactivated. To prevent this, we propose to derive $s\subscript{K}$ with a randomly selected ``litter'' sequence $l\in\mathbb{L}$ in this case.} Particularly, the set of litter codes $\mathbb{L}\subseteq\{0,1\}^n$ shall fulfill
\begin{equation}
	\not\exists \{k,l\}\in\mathbb{K}\times\mathbb{L}: D\subscript{Hamm}(c\subscript{K}(k),l)\leqslant D\subscript{max},
\end{equation}
where $D\subscript{Hamm}(x,y)$ is the Hamming distance between $x$ and $y$, and $D\subscript{max}$ is the maximal distance of a received codeword from the codebook for the channel decoder $c^{-1}\subscript{K}$ to correct errors.  The waveforming stage remains the same like with the deceptive cipherer activated.

When receiving a message $r_i$, the receiver $i$ is supposed to first decode the primary codeword $\hat{m}$ therefrom. Subsequently, it carries out \ac{sic} and try to decode a key $\hat{k}$ from the remainder, which consists of both the secondary component and the noise. \revise{}{We consider both the decoders for the ciphertext and for the key to have bounded distance, i.e., they will reject all estimates and report an erasure if the Hamming distance between the received bits and the closest codeword exceeds a certain threshold~\cite{Forney1968exponential}. Practically, bounded-distance decoders are often implemented by cascading an unbounded-distance decoder with a \ac{crc}.} If the receiver \revise{fails to decode any valid $\hat{k}$}{captures an erasure} on the second step, it perceives the situation of deactivated ciphering and takes $\hat{m}$ as the plaintext. Otherwise, it uses the decoded $\hat{k}$ to decipher $\hat{m}$ for recovery of the plaintext $\hat{p}$.

Challenging the worst case where the eavesdropper has maximum knowledge of the security measures, in this work we assume that the \revise{sets $\mathbb{P}$, $\mathbb{M}$, $\mathbb{K}$, $f$, $f^{-1}$}{tuple $\left(\mathbb{P}, \mathbb{M}, \mathbb{K}, f, f^{-1}\right)$}, as well as the modulation and channel coding schemes, are all \emph{common knowledge} shared among \emph{Alice}, \emph{Bob}, and \emph{Eve}. In this case, both \emph{Bob} and \emph{Eve} are capable of attempting with \ac{sic} to sequentially decode $m$ and $k$ from received signals. We assume that both \emph{Bob} and \emph{Eve} have perfect knowledge of their own channels, so that ideal channel equalization is achieved by both. \revise{}{Note that this setup:
\begin{enumerate}
	\item Reflects a special kind of practical security risk, such like compromised database confidentiality, or a malicious insider, where the eavesdropper has full access to the same knowledge about system as the legitimate receiver.
	\item Excludes computational security aspects, which are not the focus of our work, leaving the security performance of the system solely dependent on the physical layer security measures.
	\item Outlines a worst-case scenario with ``ideal'' eavesdropper, which provides a lower bound of the system's security performance.
\end{enumerate}
}
\subsection{Error Model}\label{subsec:error_model}
With the deceptive ciphering activated, for both the ciphertext $m$ and the key $k$, there can be three different results of the decoding:
\begin{enumerate}
	\item \emph{Success}: when the bit errors are within the error correction capability of the channel decoder, the data is correctly obtained.
	\item \emph{\revise{Failure}{Erasure}}: when the bit errors exceed the receiver's error correction capability, but \revise{not its error detection capability}{remain within the \ac{crc}'s detection capacity}, the receiver will \revise{be aware and report a failure}{report an erasure}.
	\item \emph{\revise{Confusion}{Error}}: if the bit errors exceed the \revise{error}{\ac{crc}'s} detection capability, the receiver will mistakenly decode the data with a wrong one, leading to \revise{a confusion}{an undetected error~\cite{Forney1968exponential}}.
\end{enumerate}
Generally, upon the combination of the decoding results of $m$ and $k$, there can be three different outcomes of the plaintext recovery, as shown in \revise{Tab.}{Table}~\ref{subtab:err_model_comprehensive}:
\begin{enumerate}
	\item \emph{Perception}: if both $m$ and $k$ are successfully decoded, the plaintext $p$ is correctly perceived.
	\item \emph{Deception}: between $m$ and $k$, in case only one is successfully decoded and the other \revise{confused}{incorrectly}, or when both of them are confused with \revise{wrong data}{undetected errors}, the receiver, not aware of the \revise{confusion}{error}, will try to recover $p$ with an incorrect $(m,k)$ pair, and thus obtain a wrong plaintext.
	\item \emph{Loss}: when either $m$ or $k$ \revise{fails to be decoded}{is erased}, the receiver is unable to obtain a valid $(m,k)$ pair to decipher with, so that the plaintext $p$ is lost. Indeed, since $m$ is always first decoded as the primary component, a failure in its decoding will automatically terminate the \ac{sic} process, and thus the decoding of $k$ will also fail.
\end{enumerate}
\revise{}{Note that the cases of {deception} and {loss} can be understood from a semantic perspective as {error} and {erasure} over the semantic channel, respectively~\cite{HZS+2024semantic}.	} In practical deployment, when \emph{Alice} is appropriately specified to encode both $m$ and $k$ with sufficient redundancies, and transmit them both with sufficient power, confusion is unlikely to happen. Thus, the error model simplifies to \revise{Tab.}{Table}~\ref{subtab:err_model_static}, and deception will be eliminated from the system.

However, if the cipherer is \emph{randomly} activated on selected messages (e.g., the most confidential ones), the deception can be reintroduced back to the system. More specifically, it is involved with the case where the receiver successfully decodes $m$, but cannot obtain any valid $k$ from the remainder signal. In this situation, the receiver, unaware of the cipherer activation status, cannot distinguish if it is due to a transmission error, or if the cipherer is deactivated (so that no key is transmitted at all but only a litter sequence). Thus, when mistaking the former case for the latter, the receiver will take the ciphertext as an unciphered plaintext, and therefore undergo a deception. In this context, the error model is shown in \revise{Tab.}{Table}~\ref{subtab:err_model_dynamic}.

\begin{table}[!htbp]
	\centering
	\caption{Reception error models of the proposed approach, with \subref{subtab:err_model_comprehensive} generic conditions, \subref{subtab:err_model_static} sufficient redundancy, and \subref{subtab:err_model_dynamic} random cipherer activation, respectively.}
	\label{tab:merged_table}
	\begin{subtable}[t]{\linewidth}
		\centering
		\begin{tabular}{ c  c   c   c  c }
			\multicolumn{2}{c}{}	&	\multicolumn{3}{c}{\textbf{Ciphertext}}\\
			\multicolumn{2}{c}{}	&	\textit{Success}	&	\textit{\revise{Confusion}{Error}}	&	\textit{\revise{Failure}{Erasure}}	\\\hhline{*{2}~*{3}-}
			&\multicolumn{1}{r|}{\textit{Success}}	&	\multicolumn{1}{c|}{\cellcolor[gray]{0.9}Perception}	& \multicolumn{1}{c|}{\cellcolor[gray]{0.5}} 	& \multicolumn{1}{c|}{\cellcolor[gray]{0.7}}\\\hhline{*{2}~|->{\arrayrulecolor[gray]{0.5}}->{\arrayrulecolor{black}}|>{\arrayrulecolor[gray]{0.7}}->{\arrayrulecolor{black}}|}
			&\multicolumn{1}{r|}{\textit{\revise{Confusion}{Error}}}	& \multicolumn{1}{c}{\cellcolor[gray]{0.5}}	& \multicolumn{1}{c|}{\cellcolor[gray]{0.5}Deception}	& \multicolumn{1}{c|}{\cellcolor[gray]{0.7}}	\\\hhline{*{2}~|*{2}-|>{\arrayrulecolor[gray]{0.7}}->{\arrayrulecolor{black}}|}
			\multirow{-3}{*}{\rotatebox{90}{\textbf{Key}}}	&\multicolumn{1}{r|}{\textit{\revise{Failure}{Erasure}}}	& \multicolumn{2}{c}{\cellcolor[gray]{0.7}}	 & \multicolumn{1}{c|}{\cellcolor[gray]{0.7}Loss}	\\\hhline{*{2}~*{3}-}
		\end{tabular}
		\subcaption{}
		\label{subtab:err_model_comprehensive}
	\end{subtable}\\
	\begin{subtable}[t]{\linewidth}
		\centering
		\begin{tabular}{ c  c  c  c |}
			\multicolumn{2}{c}{}	&	\multicolumn{2}{c}{\textbf{Ciphertext}}\\
			\multicolumn{2}{c}{}	&	\textit{Success}	&	\multicolumn{1}{c}{\textit{\revise{Failure}{Erasure}}}	\\\hhline{*{2}~*{2}-}
			&\multicolumn{1}{r|}{\textit{Success}}	&	\multicolumn{1}{c|}{\cellcolor[gray]{0.9}Perception}	& {\cellcolor[gray]{0.7}} 	\\\hhline{*{2}~|->{\arrayrulecolor[gray]{0.7}}->{\arrayrulecolor{black}}|}
			\multirow{-2}{*}{\rotatebox{90}{\textbf{Key}}} &\multicolumn{1}{r|}{\textit{\revise{Failure}{Erasure}}}	& \multicolumn{1}{c}{\cellcolor[gray]{0.7}}	&  Loss{\cellcolor[gray]{0.7}}	\\\hhline{*{2}~|*{2}-|}
		\end{tabular}
		\subcaption{}
		\label{subtab:err_model_static}
	\end{subtable}\\
	\begin{subtable}[t]{\linewidth}
		\centering
		\begin{tabular}{ c  c  c  c |}
			\multicolumn{2}{c}{}	&	\multicolumn{2}{c}{\textbf{Ciphertext}}\\
			\multicolumn{2}{c}{}	&	\textit{Success}	&	\multicolumn{1}{c}{\textit{\revise{Failure}{Erasure}}}	\\\hhline{*{2}~*{2}-}
			&\multicolumn{1}{r|}{\textit{Success}}	&	\multicolumn{1}{c|}{\cellcolor[gray]{0.9}Perception}	& {\cellcolor[gray]{0.7}} 	\\\hhline{*{2}~|->{\arrayrulecolor[gray]{0.7}}->{\arrayrulecolor{black}}|}
			\multirow{-2}{*}{\rotatebox{90}{\textbf{Key}}} &\multicolumn{1}{r|}{\textit{\revise{Failure}{Erasure}}}	& \multicolumn{1}{c|}{\cellcolor[gray]{0.5}Deception}	& \multirow{-2}{*}{Loss}{\cellcolor[gray]{0.7}}	\\\hhline{*{2}~|*{2}-|}
		\end{tabular}
		\subcaption{}
		\label{subtab:err_model_dynamic}
	\end{subtable}
\end{table}

From the system-level perspective, the cases of deception and loss distinguish from each other regarding their utility impacts. If an incorrectly decoded message creates the same utility for the receiver as a lost message does (a zero-utility is usually considered in this case), deception will be practically equivalent to loss. Nevertheless, in many scenarios, it is possible to design the system so that a deception will lead to a significant penalty, which can be exploited to actively counter eavesdroppers. This has been studied in~\cite{HZS+2023nonorthogonal}, and we will further elaborate on the use cases in Sec.~\ref{subsec:use_cases}.

\subsection{Performance Metrices}
Conventional \ac{pls} approaches, mostly working in the \ac{ibl} regime, commonly refer to the secrecy capacity to evaluate the security performance. However, in the \ac{fbl} regime, the classical concept of channel capacity is invalid since error-free transmission is hardly achievable~\cite{YSP+2019wiretap}.
As an alternative performance indicator for secure and reliable communication, we introduce the \ac{lfp}, i.e. the probability that the plaintext is either correctly perceived by \emph{Eve}, or not perceived by \emph{Bob}:
\begin{equation}\label{eq:lfp_definition}
    \varepsilon\lf=1-(1-\varepsilon\bob)\varepsilon\eve,
\end{equation}
where $\varepsilon\bob$ and $\varepsilon\eve$ are the non-perception probabilities of \emph{Bob} and \emph{Eve}, respectively. Notating as $\varepsilon_{i,j}$ the failure probability of receiver $i\in\{\text{Bob}, \text{Eve}\}$ at decoding the message component $j\in\{\text{M}, \text{K}\}$, we have
\begin{equation}
\label{Eq: overall error rate}
    \varepsilon_i=1-(1-\varepsilon_{i,\mathrm{M}})(1-\varepsilon_{i,\mathrm{K}}),\quad\forall i \in \{\text{Bob}, \text{Eve}\},
\end{equation}
and therefore
\begin{equation}
\begin{split}
    \varepsilon\lf=&1-\left (1-\varepsilon\subscript{Bob,M} \right) \left (1-\varepsilon\subscript{Bob,K}\right)\\
	&\times\left[ 1-\left (1-\varepsilon\subscript{Eve,M} \right) \left (1-\varepsilon\subscript{Eve,K} \right)  \right]
\end{split}
\end{equation}

Additionally, to evaluate the performance of deceiving eavesdroppers, we define the effective deception rate as the probability that not \emph{Bob} but only \emph{Eve} is deceived:
\begin{equation}\label{eq:eff_deception_rate}
	R\subscript{d}=\left[1-\left(1-\varepsilon\subscript{Bob,M}\right)\varepsilon\subscript{Bob,K}\right](1-\varepsilon\subscript{Eve,M})\varepsilon\subscript{Eve,K}
\end{equation}

According to~\cite{PPV2010channel}, the error probability $\varepsilon_{i,j}$ with a given packet size $d_j$ can be written as:
\begin{equation}
\label{eq:FBL_error}
    \begin{aligned}
    \varepsilon_{i,j} &=\mathcal{P}\left(\gamma_{i,j},d_j,n\right) \\&\approx \revise{\mathcal{Q}}{Q} \left(\sqrt{\frac{n}{V(\gamma_{i,j})}}\left(\mathcal{C}(\gamma_{i,j})-\frac{d_j}{n}\right) \ln{2}\right),
    \end{aligned}
\end{equation}
where  $\revise{\mathcal{Q}}{Q}(x)=\frac{1}{\sqrt{2 \pi}} \int_{x}^{\infty} e^{-t^2/2} dt$ is the Q-function in statistic, ,$\mathcal{C}(\gamma)=\log_2(1+\gamma)$ is the Shannon capacity, ${V}(\gamma)=1-\frac{1}{{{\left( 1+{{\gamma }} \right)}^{2}}}$ is the channel dispersion. 

\subsection{Strategy Optimization}
Towards a novel reliable and secure communication solution that well counters eavesdroppers, we aim for maximizing the effective deception rate while maintaining a low \ac{lfp}. In practical wireless systems, the blocklength of codewords $n$ shall be fixed to fit the radio numerology, while the plaintext length $d\subscript{P}$ is determined by the application requirements, leaving us only three degrees of freedom to optimize: the key length $d\subscript{K}$, and the powers $\left(P\subscript{M}, P\subscript{K}\right)$ that are allocated to the message and the key components, respectively. Thus, the original optimization problem \reviseRT{$(OP)$}{} can be formulated as follows:
\begin{maxi!}
	{d\subscript{K}, P\subscript{M}, P\subscript{K} }{R\subscript{d}\label{obj:deception_rate}} {\label{op:origin}}{\reviseRT{(OP):}{}}{}
	\addConstraint{P\subscript{M}\geqslant 0 }
	\addConstraint{P\subscript{K}\geqslant 0 \label{con:non-neg_power_key}}
	\addConstraint{P\subscript{M}+P\subscript{K}\leqslant P_\Sigma \label{con:max_power}}
	\addConstraint{d\subscript{K}\in\{0,1,\dots n\} \label{con:max_key_length}}
	\addConstraint{\varepsilon\subscript{Bob,M}\leqslant \varepsilon\superscript{th}\subscript{Bob,M} \label{con:err_bob_message_threshold}}
	\addConstraint{\varepsilon\subscript{Eve,M}\leqslant \varepsilon\superscript{th}\subscript{Eve,M} \label{con:err_eve_message_threshold}}
	\addConstraint{\varepsilon\subscript{Bob,K}\leqslant \varepsilon\superscript{th}\subscript{Bob,K} \label{con:err_bob_key_threshold}}
	\addConstraint{\varepsilon\subscript{Eve,K}\geqslant \varepsilon\superscript{th}\subscript{Eve,K} \label{con:err_eve_key_threshold}}
    \addConstraint{\varepsilon\lf\leqslant \varepsilon\superscript{th}\lf, \label{con:LFP_threshold}}
\end{maxi!}
where $\varepsilon\superscript{th}\subscript{Bob,M}$, $\varepsilon\superscript{th}\subscript{Eve,M}$, $\varepsilon\superscript{th}\subscript{Bob,K}$, $\varepsilon\superscript{th}\subscript{Eve,K}$, and $\varepsilon\superscript{th}\subscript{LF}$ are pre-determined thresholds. \revise{}{Moreover, $P_\Sigma$ is the total power threshold}.

\section{Proposed Approach}\label{sec:approach}
\subsection{Analyses}\label{subsec:analyses}
Due to the non-convexity of the deception rate $R\subscript{d}$, Problem~\eqref{op:origin} is challenging to solve, and thus, in this section, we first reformulate the original problem to an equivalent, yet simpler one. With our analytical findings, we establish the partial convexity feature of the objective function with respect to each optimization variable. 

In particular, we first relax $d\subscript{K}$ from integer to real value, i.e., $d\subscript{K}\in \mathbb{N} \to d\subscript{K}\in \mathbb{R}_+$. Then, we establish the following theorem to characterize the optimal condition of Problem~\eqref{op:origin}:

\begin{theorem}
\label{theorem:max_power}
    Given any $d\subscript{K}\geqslant 0$, the optimal power allocation must fulfill $P\superscript{o}\subscript{M}+P\superscript{o}\subscript{K}=P_\Sigma$.
\end{theorem}
\begin{proof}
    See Appendix~\ref{proof:max_power}.
\end{proof}

Theorem~\ref{theorem:max_power} indicates that the inequality constraint~\eqref{con:max_power} can be turned into an equality constraint. Therefore, we can eliminate $P\subscript{K}$ with $P\subscript{K}=P_{\Sigma}-P\subscript{M}$  without affecting the optimality of the solutions for Problem~\eqref{op:origin}. Furthermore, since $R\subscript{d}$ is always positive and non-zero, to maximize it is to minimize its multiplicative inverse, i.e., we have the following equivalent problem:

\begin{mini!}
	{d\subscript{K},  P\subscript{M} }{\frac{1}{R\subscript{d}}\label{obj:sca}} {\label{Problem_SCA}}{\reviseRT{(SP):}{}}{}
	\addConstraint{P\subscript{M}\geqslant 0 \label{con:non-neg_power_message}}
	\addConstraint{P\subscript{M}+P\subscript{K}= P_\Sigma \label{con:max_power_equal}}
	\addConstraint{0\leqslant d\subscript{K} \leqslant n \label{con:max_key_length_relax}}
    \addConstraint{\eqref{con:err_bob_message_threshold} - \eqref{con:LFP_threshold}\label{con:repeated_cons_sca}}
\end{mini!}
However, due to the multiplication of error probabilities in~\eqref{eq:eff_deception_rate}, it is still a non-convex problem. Therefore, we provide the following lemma to decouple it:
\begin{lemma}
\label{lemma:approximation}
    Given any local point $\left(\hat{d}\subscript{K}^{(q)},\hat{P}\subscript{M}^{(q)}\right)$, $\frac{1}{R\subscript{d}}$ is upper-bounded by an approximation $\hat{R}\subscript{d}\left(d\subscript{K},P\subscript{M}\left\vert \hat{d}\subscript{K}^{(q)},\hat{P}\subscript{M}^{(q)}\right.\right)$, i.e.,
    \begin{equation}
    \label{eq:R_app}
    \begin{split}
        \frac{1}{R\subscript{d}(d\subscript{K},P\subscript{M})}&=\frac{1}{R_b(1-\varepsilon\subscript{Eve,M})\varepsilon\subscript{Eve,K}}\\&\leqslant 
        \frac{1}{9\lambda_1\left(\hat{d}\subscript{K}^{(q)},\hat{P}\subscript{M}^{(q)}\right)\lambda_2\left(\hat{d}\subscript{K}^{(q)},\hat{P}\subscript{M}^{(q)}\right)}\\
        &~~~~~~~~~~~~\cdot \left(\frac{1}{R_b}+\frac{\lambda_1}{(1-\varepsilon\subscript{Eve,M})}+\frac{\lambda_2}{\varepsilon\subscript{Eve,K}}                 
            \right)^3\\
        &\triangleq \hat{R}\subscript{d}\left(d\subscript{K},P\subscript{M}\left\vert \hat{d}\subscript{K}^{(q)},\hat{P}\subscript{M}^{(q)}\right.\right),
    \end{split}
    \end{equation}
    where $R_b=1-(1-\varepsilon\subscript{Bob,M})\varepsilon\subscript{Bob,K}$. Moreover, 
    \begin{equation}
    \lambda_1\left(\hat{d}\subscript{K}^{(q)},\hat{P}\subscript{M}^{(q)}\right)=\frac{R_b\left(\hat{d}\subscript{K}^{(q)},\hat{P}\subscript{M}^{(q)}\right)}{1-\varepsilon\subscript{Eve,M}\left(\hat{d}\subscript{K}^{(q)},\hat{P}\subscript{M}^{(q)}\right)},
    \end{equation} 
    
    \begin{equation}
    \lambda_2\left(\hat{d}\subscript{K}^{(q)},\hat{P}\subscript{M}^{(q)}\right)=\frac{R_b\left(\hat{d}\subscript{K}^{(q)},\hat{P}\subscript{M}^{(q)}\right)}{\varepsilon\subscript{Eve,K}\left(\hat{d}\subscript{K}^{(q)},\hat{P}\subscript{M}^{(q)}\right)},
    \end{equation} 
    are non-negative constants at the local point $\left(\hat{d}\subscript{K}^{(q)},\hat{P}\subscript{M}^{(q)}\right)$.  
    \label{lemma:convexity}
\end{lemma}
\begin{proof}
    See Appendix~\ref{proof:approximation}. 
\end{proof}

Interestingly, $R\subscript{d}$ is equal to its upper-bound $\hat{R}\subscript{d}$ at the local point, i.e., $\frac{1}{R\subscript{d}}\left(\hat{d}\subscript{K}^{(q)},\hat{P}\subscript{M}^{(q)}\right)=\hat{R}\subscript{d}\left(\hat{d}\subscript{K}^{(q)},\hat{P}\subscript{M}^{(q)}\left\vert \hat{d}\subscript{K}^{(q)},\hat{P}\subscript{M}^{(q)}\right.\right)$. This observation inspires us to leverage the \ac{mm} algorithm~\cite{HL2000MM} combing with the \ac{bcd} method~\cite{tseng2001BCD} to solve the optimization problem. 

In order to do so, we still need further modifications to \reviseRT{$(SP)$}{Problem~\eqref{Problem_SCA}}. In particular, with any local point $(\hat{d}\subscript{K}^{(q)},\hat{P}\subscript{M}^{(q)})$ and the corresponding approximation $\hat{R}\subscript{d}\left(\hat{d}\subscript{K}^{(q)},\hat{P}\subscript{M}^{(q)}\left\vert \hat{d}\subscript{K}^{(q)},\hat{P}\subscript{M}^{(q)}\right.\right)$, we decompose the problem in each $t\superscript{th}$ iteration by letting $P\subscript{M}$ to be a fixed $P^{(t)}\subscript{M}$. The corresponding problem is given by:



\begin{mini!}
	{d\subscript{K}}{\hat{R}\subscript{d}^{(t)}\left(d\subscript{K}\left\vert \hat{d}\subscript{K}^{(q)},\hat{P}\subscript{M}^{(q)}\right.\right)\label{obj:sca_bcd_d_K}} {\label{Problem_SCA_BCD_d_K}}{\reviseRT{(SP1):}{}}{}
	\addConstraint{P\subscript{M}=P\subscript{M}^{(t)}}
	\addConstraint{0\leqslant d\subscript{K} \leqslant n}    \addConstraint{\eqref{con:err_bob_message_threshold} - \eqref{con:LFP_threshold}.\label{con:repeated_cons_sca_bcd_d_K}}
\end{mini!}
Therefore, \reviseRT{$(SP1)$}{Problem~\eqref{Problem_SCA_BCD_d_K}} becomes a single-variable problem and we have the following lemma to characterize it:
\begin{lemma}
\label{lemma:convex_d_K}
    Problem~\eqref{Problem_SCA_BCD_d_K} is convex.
\end{lemma}
\begin{proof}
    See Appendix~\ref{proof:convex_d_K}.
\end{proof}

According to Lemma~\ref{lemma:convex_d_K}, we can solve Problem~\eqref{Problem_SCA_BCD_d_K} with optimal solution $d^\circ\subscript{K}$ efficiently via any standard convex programming. On the other hand, we have the second decomposed problem in the $t\superscript{th}$ inner iteration by letting $d\subscript{K}=d\subscript{K}^{(t)}$:
\begin{mini!}
	{P\subscript{M}}{\hat{R}\subscript{d}^{(t)}\left(d\subscript{K}\left\vert \hat{d}\subscript{K}^{(q)},\hat{P}\subscript{M}^{(q)}\right.\right)\label{obj:sca_bcd_P_M}} {\label{Problem_SCA_BCD_P_M}}{\reviseRT{(SP2):}{}}{}
	\addConstraint{0\leqslant P\subscript{M} \leqslant P_{\Sigma}}
	\addConstraint{d\subscript{K}=d_{K}^{(t)}}    \addConstraint{\eqref{con:err_bob_message_threshold} - \eqref{con:LFP_threshold}\label{con:repeated_cons_sca_bcd_P_M}}
\end{mini!}
Similarly, we have the following analytical find:
\begin{theorem}
    Problem~\eqref{Problem_SCA_BCD_P_M} is convex.
    \label{theorem:convex_P_M}
\end{theorem}
\begin{proof}
    See Appendix~\ref{proof:convex_P_M}.
\end{proof}
Therefore, we can also solve Problem~\eqref{Problem_SCA_BCD_P_M} with optimal solution $P^\circ\subscript{M}$ efficiently via convex programming. Let $P^{(t+1)}\subscript{M}=P^\circ\subscript{M}$.


\subsection{Optimization Algorithm}\label{subsec:algorithm}
With the above analyses, we propose an efficient algorithm with two layers of iterations to obtain the solutions of \reviseRT{$(OP)$}{Problem~\eqref{op:origin}}. 
In particular, in each $q\superscript{th}$ outer iteration, we approximate $\frac{1}{R\subscript{d}}$ with \revise{$\hat{R}\subscript{d}^{(q)}:=\hat{R}\subscript{d}\left(\hat{d}\subscript{K}^{(q)},\hat{P}\subscript{M}^{(q)}\left\vert \hat{d}\subscript{K}^{(q)},\hat{P}\subscript{M}^{(q)}\right.\right)$}{$\frac{1}{R\subscript{d}}$ with $\hat{R}\subscript{d}^{(q)}:=\hat{R}\subscript{d}\left(\hat{d}\subscript{K}^{(q-1)},\hat{P}\subscript{M}^{(q-1)}\right)$}. Then, in each $t\superscript{th}$ inner iteration, we solve \reviseRT{$(SP1)$}{Problem~\eqref{Problem_SCA_BCD_d_K}} as a single-variable convex problem by letting $P\subscript{M}=P\subscript{M}^{(t)}$. Denote its optimal solution as $d\subscript{K}^{(t)}$. We solve \reviseRT{$(SP1)$}{Problem~\eqref{Problem_SCA_BCD_P_M}} also as a single-variable convex problem by letting $P\subscript{M}=P\subscript{M}^{(t)}$. Denote its optimal solution as $P\subscript{M}^{(t+1)}$ and enter the next $(t+1)\superscript{th}$ inner iteration. This process is repeated until the stop criterion $\left\vert\hat{R}\subscript{d}^{(t)}\left(\left.d\subscript{K}^{(t)},P\subscript{M}^{(t)}\right\vert\hat{d}\subscript{K}^{(q)},P\subscript{M}^{(q)}\right)\right.-\left.\hat{R}\subscript{d}^{(t-1)}\left(\left.d\subscript{K}^{(t-1)},P\subscript{M}^{(t-1)}\right\vert\hat{d}\subscript{K}^{(q)},P\subscript{M}^{(q)}\right)\right\vert\leqslant \mu\subscript{BCD}$ meets, where $\mu\subscript{BCD}$ is the stop threshold of the inner iteration. Then, we assign $\left(d\subscript{K}^{(t)},P\subscript{M}^{(t)}\right)=\left(\hat{d}\subscript{K}^{(q+1)},\hat{P}\subscript{M}^{(q+1)}\right)$ as the local point of the next $(q+1)\superscript{th}$ outer iteration and approximate the objective function of \reviseRT{$(OP)$}{Problem ~\eqref{op:origin}} with $\hat{R}\subscript{d}^{(q+1)}$. The inner iteration is reset with $t=1$ and start again. This process is  repeated until the stop criterion $\left\vert\hat{R}\subscript{d}^{(t)}\left(\left.d\subscript{K}^{(t)},P\subscript{M}^{(t)}\right\vert\hat{d}\subscript{K}^{(q)},P\subscript{M}^{(q)}\right)\right.-\left.\hat{R}\subscript{d}^{(t)}\left(\left.d\subscript{K}^{(t)},P\subscript{M}^{(t)}\right\vert\hat{d}\subscript{K}^{(q-1)},P\subscript{M}^{(q-1)}\right)\right|\leqslant \mu\subscript{MM}$ meets, where $\mu\subscript{MM}$ is the stop threshold of outer iteration. The obtained solution is denoted as $P\subscript{K}\superscript{*}$ and $d\subscript{K,R}\superscript{*}$. 
Specially, we let $P\subscript{M}\superscript{(0)}=P\subscript{M}\superscript{init}, d\subscript{K}\superscript{(0)}=d\subscript{K}\superscript{init} ,
\hat{P}\subscript{M}\superscript{(0)}=\hat{P}\subscript{M}\superscript{init}, \hat{d}\subscript{K}\superscript{(0)}=\hat{d}\subscript{K}\superscript{init} ,
R\subscript{d}\superscript{(0)}=-\infty$ in the initial round of the iteration. 
It is important to note that the initial values must be feasible for \reviseRT{$(OP)$}{Problem~\eqref{op:origin}}. Remembering that $d\subscript{K}$ must be an integer, the optimal integer solution will be determined by comparing the integer neighbors of $d\subscript{K,R}\superscript{*}$:
\begin{equation}
    d\subscript{K}\superscript{*}=\arg \underset{m \in \left\{\left\lfloor d\subscript{K,R}\superscript{*} \right\rfloor, \left\lceil d\subscript{K,R}\superscript{*} \right\rceil\right\}}{\max} R\subscript{d}\left(P\subscript{M}\superscript{*}\right).
\end{equation}

Clearly, the inner iteration is a \ac{bcd} methods with decomposed sub-problems and the outer iteration is a \ac{mm} algorithm with successive convex approximations. This approach to solve Problem~\eqref{op:origin} is described in Algorithm~\revise{}{\ref{alg:framework}}. \revise{}{The method can attain near-optimal solutions with a complexity of $\mathcal{O}\left(\phi\left(4N^2\right)\right)$, where $N$ denotes the number of variables in Problem~\eqref{op:origin} and $\phi(\cdot)$ signifies the number of iterations based on the accuracy of the solution.}

\begin{algorithm}[!htbp]
	\SetAlgoLined
	Input: $\mu\subscript{BCD},\mu\subscript{MM}, T,Q,P\subscript{\Sigma}, d\subscript{M}, n, $ \\
	Initialize: $t=1, q=1, P\subscript{M}\superscript{(0)}=P\subscript{M}\superscript{init}, d\subscript{K}\superscript{(0)}=d\subscript{K}\superscript{init} ,
	\hat{P}\subscript{M}\superscript{(0)}=\hat{P}\subscript{M}\superscript{init}, \hat{d}\subscript{K}\superscript{(0)}=\hat{d}\subscript{K}\superscript{init} ,
	R\subscript{d}\superscript{(0)}=-\infty$ \\
	\Do{
		$\frac{\hat{R}^{(t)}\subscript{d}\left(\hat{d}^{(q)}\subscript{K},\hat{P}^{(q)}\subscript{M}\right)-\hat{R}^{(t)}\subscript{d}\left(\hat{d}^{(q-1)}\subscript{K},\hat{P}^{(q-1)}\subscript{M}\right)}{\hat{R}^{(t)}\subscript{d}\left(\hat{d}^{(q-1)}\subscript{K},\hat{P}^{(q-1)}\subscript{M}\right)}> \mu\subscript{MM}$
		}{
		\eIf{
			$q\leqslant Q$
		}{
			$t\leftarrow 1$ (reset index $t$)\\
			$\hat{R}^{(q)}\subscript{d}:= \hat{R}\subscript{d}\left(\hat{d}^{(q-1)}\subscript{K},\hat{P}^{(q-1)}\subscript{M}\right)$ \\
			\Do{
				$\frac{\hat{R}\subscript{d}^{(t)}-\hat{R}\subscript{d}^{(t-1)}}{\hat{R}\subscript{d}^{(t-1)}}> \mu\subscript{BCD}$
				}{
				\eIf{
					$t\leqslant T$
				}{
					$P\subscript{M}^{(t)}\leftarrow \arg \underset{P\subscript{M}}{\min}\; \hat{R}\subscript{d}\left(d\subscript{K}^{(t-1)},P\subscript{M}\right)$ \\
					$d\subscript{K}^{(t)}\leftarrow \arg \underset{d\subscript{K}}{\min}\; \hat{R}\subscript{d}\left(d\subscript{K},P\subscript{M}^{(t)}\right)$ \\
					$\hat{R}\subscript{d}^{(t)}\leftarrow \hat{R}\subscript{d}\left(d\subscript{K}^{(t)},P\subscript{M}^{(t)}\right)$ \\
					$t \leftarrow t+1$
				}{
				\textbf{break}
				}
			}
			$\hat{d}\subscript{K}^{(q)}\leftarrow d\subscript{K}^{(t)}$\\ 
			$\hat{P}\subscript{K}^{(q)}\leftarrow P\subscript{K}^{(t)}$\\
			$q\leftarrow q+1$
		}{
		\textbf{break}
		}
	}
	$P\subscript{K}\superscript{*}\leftarrow P\subscript{K}^{(q)}$ \\
	$d\subscript{K,R}\superscript{*}\leftarrow \arg \underset{m \in \left\{\left\lfloor d\subscript{K}^{(q)} \right\rfloor, \left\lceil d\subscript{K}^{(q)} \right\rceil\right\}}{\max} R\subscript{d}\left(P\subscript{M}^{(q)}\right)$ \\
	\KwRet{$\left(d\subscript{K}\superscript{*}, P\subscript{M}\superscript{*}\right)$}
\caption{The \revise{Proposed}{proposed \ac{mm}-\ac{bcd}} framework}
\label{alg:framework}
\end{algorithm}

\revise{}{\subsection{Real-Time Adaptation to Channel Dynamics}\label{subsec:ch_dynamics}
So far in this work, we have only discussed the static scenario where both $z\subscript{Bob}$ and $z\subscript{Eve}$ remain consistent. However, in practical scenarios, the channel gains are subject to high dynamics due to factors such like user mobility, environment changes, and interference from other devices.

Therefore, towards practical deployment, the \ac{pld} approach must be able to adapt to the channel dynamics in real-time. This can be achieved by periodic channel measurement to update the \ac{csi} within channel coherence time, and therewith re-optimize the power allocation and key length. 
While online solution of the optimization problem can be computationally expensive, a practical and cost-efficient approach is to pre-calculate a set of optimal solutions for different channel conditions, and select the most suitable one from a \ac{lut} w.r.t. the real-time channel measurement. 
}

\section{Numerical Evaluation}\label{sec:evaluation}
To validate our theoretical analyses and evaluate the proposed approach, a series of numerical experiments were conducted. The common parameters of the simulation setup are listed in \reviseRT{Tab.}{Table~}\ref{tab:sim_setup}, while task-specific ones will be detailed later.

\begin{table}[!htpb]
	\centering
	\caption{Simulation setup}
	\label{tab:sim_setup}
	\begin{tabular}{m{1.2cm} m{1.5cm} m{4.3cm}}
		\toprule[2px]
		\textbf{Parameter} 		&\textbf{Value} 		&\textbf{Remark}\\
		\midrule[1px]
		$\sigma^2$					&\SI{1}{\milli\watt}			&Noise power\\
        \rowcolor{gray!30}
		$z\bob$					  &\SI{0}{\dB}				&Channel gain of \emph{Bob}\\
		$B$								 &\SI{1}{\hertz}			&Normalized to unity bandwidth\\
		\rowcolor{gray!30}
		$n$								 &64							 &Block length per packet\\
		$\varepsilon\subscript{Bob,M}\superscript{th}$ $\varepsilon\subscript{Bob,K}\superscript{th}$ $\varepsilon\subscript{Eve,M}\superscript{th}$ $\varepsilon\subscript{Eve,K}\superscript{th}$		&0.5							 &Thresholds in constraints \eqref{con:err_bob_message_threshold}--\eqref{con:err_eve_key_threshold}\\
        \rowcolor{gray!30}
		\revise{}{$\mu\subscript{MM}$}							&\revise{}{$1\times 10^{-7}$}	&\revise{}{\ac{mm} convergence threshold}\\
		\revise{$\xi$}{$\mu\subscript{BCD}$}							&\revise{}{$1.49\times 10^{-8}$}	&\ac{bcd} convergence threshold\\
		\rowcolor{gray!30}
		\revise{}{$Q$}				&\revise{}{100}							&\revise{}{Maximal number of iterations in \ac{mm}}\\
		$K$								&100							&Maximal number of iterations in \ac{bcd}\\
		\bottomrule[2px]
	\end{tabular}
\end{table}

\subsection{Superiority of Full-Power Transmission}\label{subsec:full_power}
To verify that the optimal power allocation scheme always fully utilizes the transmission power budget, we set $P\subscript{\Sigma}=\SI{10}{\milli\watt}$, $z\eve=-\SI{10}{\dB}$, $d\subscript{M}=16$, and calculated the deception rate $R\subscript{d}$ according to \reviseRT{Eq.}{Problem}~\eqref{op:origin} in the region $(P\subscript{M}, P\subscript{K}) \in \left[0, \SI{10}{\milli \watt}\right]^2$. We then performed exhaustive search to find the optimal $P\subscript{M}\superscript{o}$ that maximizes $R\subscript{d}$ in the feasible region of Problem~\eqref{op:origin} with $\varepsilon\superscript{th}\lf=0.5$, for two different cases where $d\subscript{K}=30$ and $d\subscript{K}=60$, respectively.

The results are illustrated in Fig.~\ref{fig:full-power}\revise{,}{. Under both setups, we can see that all the optimal combinations are located on the line $P\subscript{M}+P\subscript{K}=P\subscript{\Sigma}$}, which confirms our theoretical analysis \revise{under both setups.}{that the optimal power allocation scheme always fully utilizes the transmission power budget.} 

\begin{figure}[!htpb]
	\centering
	\includegraphics[width=\linewidth]{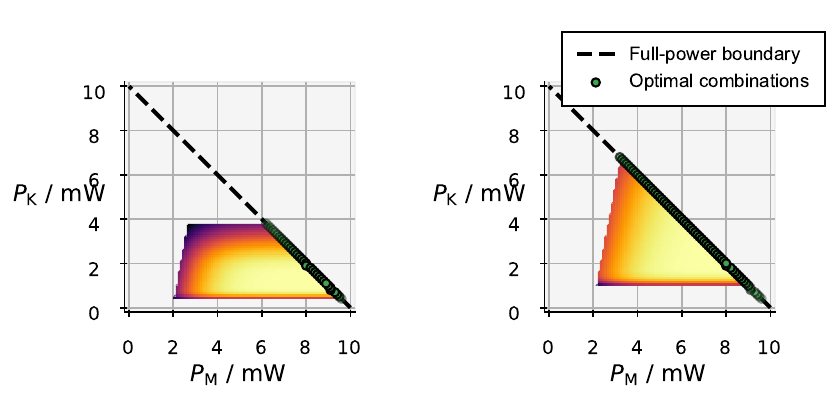}
	\caption{The optimal combinations of $P\subscript{M}$ and $P\subscript{K}$ in case of $P\subscript{\Sigma}=\SI{10}{\milli\watt}$, with $d\subscript{K} = 30$ (left) and $d\subscript{K} = 60$ (right). }
	\label{fig:full-power}
\end{figure}

\subsection{Deception Rate Surface}\label{subsec:deception_rate_surface}
To gain insight into the deception rate surface $R\subscript{d}$ under the full-power transmission scheme, we set $P\subscript{\Sigma}=\SI{10}{\milli \watt}$, $z\eve=\SI{-10}{\dB}$, $d\subscript{M}=16$, and computed $R\subscript{d}$ in the region $(P\subscript{M}, d\subscript{K}) \in [0, \SI{10}{\milli \watt}]\times \{0,1,...64\}$ with $\varepsilon\superscript{th}\lf=0.5$. The result is illustrated in Fig. \ref{fig:deception_rate-0.5}, where the feasible region defined by constraints \eqref{con:err_bob_message_threshold}--\eqref{con:LFP_threshold} is highlighted with greater opacity compared to the rest.

From this figure, we observe that within the feasible region, the deception rate $R\subscript{d}$ exhibits concavity with respect to both \revise{$P\subscript{M}$ and $d\subscript{K}$}{$d\subscript{K}$ and $P\subscript{M}$, which we have analytically derived as Lemma~\ref{lemma:convex_d_K} and Theorem~\ref{theorem:convex_P_M}, respectively}. However, the behavior regarding convexity or concavity outside this region appears to be more complex.

\begin{figure}[!htpb]
	\centering
	\includegraphics[width=\linewidth]{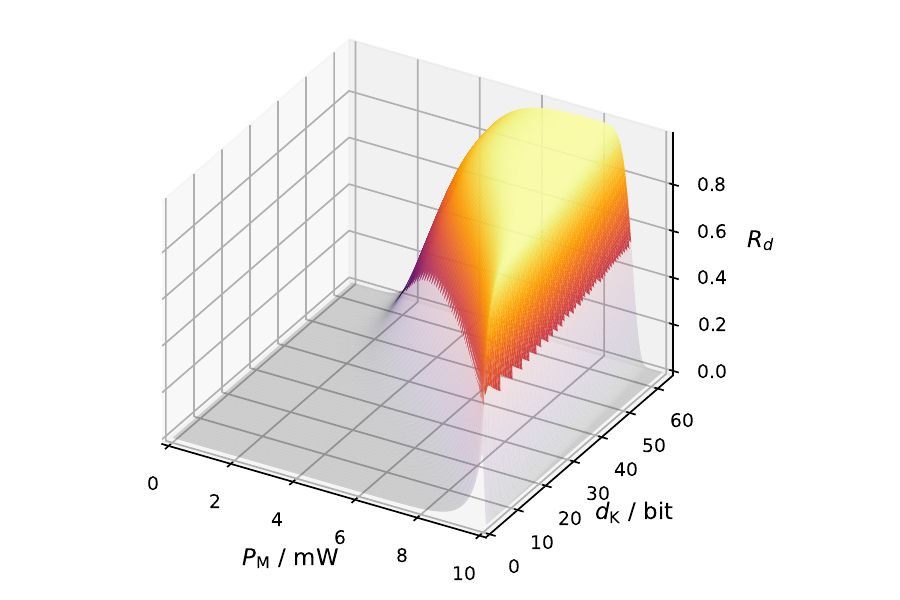}
	\caption{Deception rate under full-power transmission with $\varepsilon\superscript{th}\lf=0.5.$}
	\label{fig:deception_rate-0.5}
\end{figure}

\subsection{Convergence Test of the Optimization Algorithm}\label{subsec:convergence_test}
To assess the practicality of the proposed \revise{\ac{bcd} algorithm}{\ac{mm}-\ac{bcd} framework} in optimizing both key length and power allocation, we conducted Monte-Carlo simulations with $P\subscript{\Sigma}=\SI{10}{\milli \watt}$, $z\eve=\SI{-10}{\dB}$, and $\varepsilon\superscript{th}\lf=0.5$. The algorithm was evaluated with two different lengths of the payload message: $d\subscript{M}=16$ and $d\subscript{M}=24$ respectively.

The results presented in Fig. \ref{fig:bcd_convergence} indicate that the \revise{\ac{bcd} algorithm}{\ac{mm}-\ac{bcd} framework} effectively reaches convergence in both cases, obtaining the optimum after 7 and 8 iterations, respectively. From this figure, it is observed that there is a tiny gap between the local optimum obtained by the \revise{\ac{bcd} algorithm}{\ac{mm}-\ac{bcd} framework} and the global optimum found by exhaustive search, which is attributed to the flatness of the region surrounding the optimal point. The error between the final step of the \revise{\ac{bcd} algorithm}{\ac{mm}-\ac{bcd} framework} and the global optimum is \revise{}{$1.85 \times 10^{-8}$} for $d\subscript{M}=16$, and \revise{}{$3.44 \times 10^{-8}$} for $d\subscript{M}=24$.

\begin{figure}[!htpb]
	\centering
	\includegraphics[width=\linewidth]{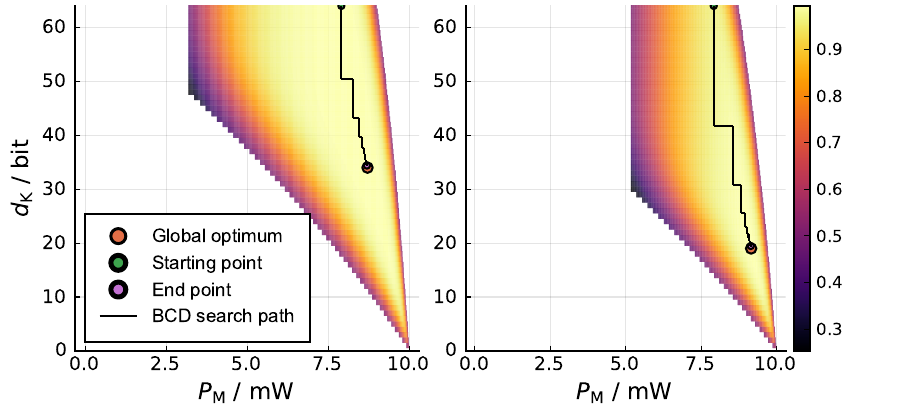}
	\caption{The $R\subscript{d}$ surface and the search path, with $d\subscript{M}=16$ (left) and $d\subscript{M}=24$ (right).}
	\label{fig:bcd_convergence}
\end{figure}

\subsection{Performance Evaluation}\label{subsec:performance_evaluation}
To evaluate the secrecy and deception performances of our proposed approach, we focus on the \ac{lfp} $\varepsilon \lf$ and the effective deception rate $R\subscript{d}$, respectively.

First, we set $P_{\Sigma}=\SI{3}{\milli \watt}$ and $\varepsilon\superscript{th}\lf=0.5$, then evaluated our method under various eavesdropping channel conditions. For benchmarking purpose, we also measured the \ac{lfp} of two conventional \ac{pls} approaches as baselines. Both the baseline solutions apply no deceptive ciphering ($d\subscript{K} = 0$, $P\subscript{K} = 0$), so they are incapable of deceiving but only minimizing $\varepsilon \lf$. The first baseline selects the optimal $P\subscript{M} \in \left[0, P\subscript{\Sigma}\right]$ regarding a fixed $d\subscript{M} = 16$, while the second searches for the best $d\subscript{M} \in [16,64]$ for a full-power transmission $P\subscript{M}=P\subscript{\Sigma}$.

The results are displayed in Fig. \ref{fig: sensitivity-h_eve}. Our \ac{pld} solution is able to maintain a satisfactory \ac{lfp} that is significantly lower than the preset threshold $\varepsilon\subscript{LF}\superscript{th}$, while exhibiting a high effective deception rate. Especially, our method is not only robust to the eavesdropping channel gain, but even slightly benefiting from a reasonably good $z\eve$. In contrast, both baselines, performing closely to each other, logarithmically increase in \ac{lfp} as $z\eve$ increases. This allows our \ac{pld} solution to outperform the baselines by a significant \ac{lfp} margin over good eavesdropping channels, while simultaneously delivering an excellent deception rate up over $95\%$. On the other hand, when the eavesdropping channel gain is poor, our method is still well capable of deceiving \emph{Eve} with $R\subscript{d}>75\%$, at only a reasonable cost of increased \ac{lfp} regarding the baselines. 

\begin{figure}[!htpb]
	\centering
	\includegraphics[width=\linewidth]{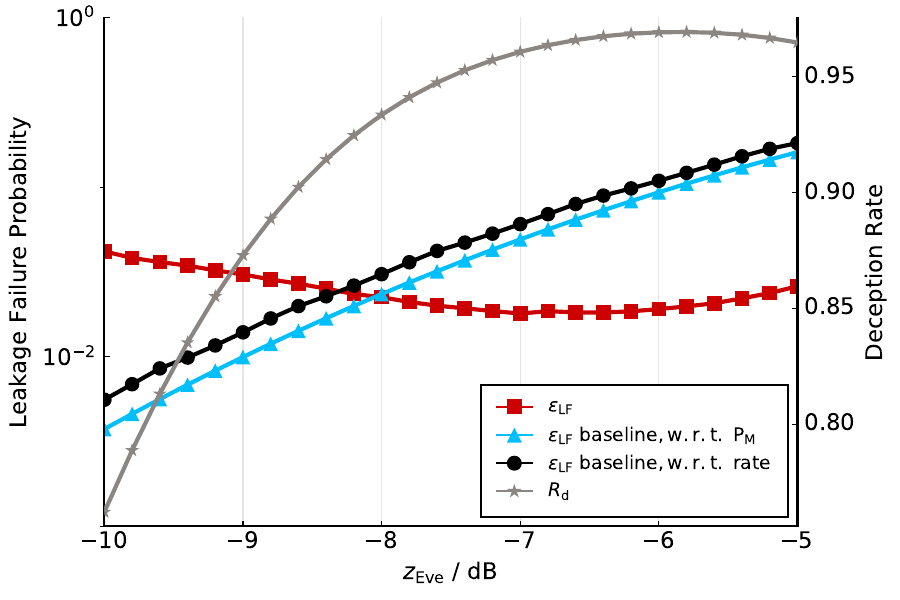}
	\caption{Results of sensitivity test regarding $z\eve$.}
	\label{fig: sensitivity-h_eve}
\end{figure}

Next, we set $z\eve= -\SI{5}{\dB}$, $\varepsilon\superscript{th}\lf=0.5$ and evaluated our method under varying power budgets $P\subscript{\Sigma}$. The results are shown in Fig. \ref{fig: sensitivity-p_total}. With an adequate power budget, the $\varepsilon \lf$ decreases notably compared to the baseline, along with a high effective deception rate. Additionally, unlike traditional \ac{pls} solutions, which do not benefit or benefit little from increased power budgets, our method performs significantly better by increasing $P\subscript{\Sigma}$.

\begin{figure}[!htpb]
	\centering
	\includegraphics[width=\linewidth]{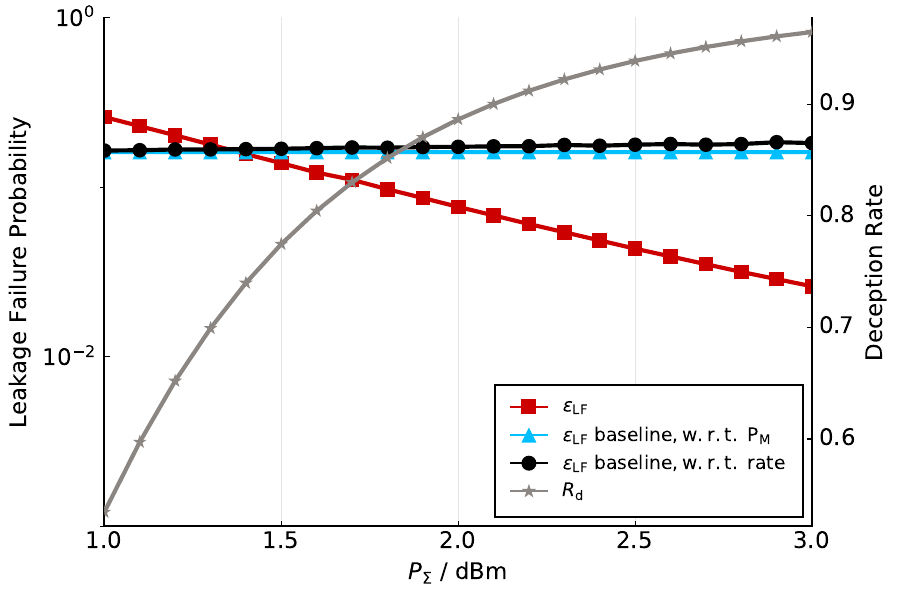}
	\caption{Results of sensitivity test regarding $P_{\Sigma}$.}
	\label{fig: sensitivity-p_total}
\end{figure}

The outcomes of a more comprehensive benchmark test, which combines various $z\eve$ and $P\subscript{\Sigma}$, are depicted in Fig. \ref{fig: benchmark results}. We still kept the setup $\varepsilon\superscript{th}\lf=0.5$. These results demonstrate that our method generally outperforms the classical \ac{pls} baselines with an adequate power budget under various eavesdropping channel conditions. In more detail, the minimum $P\subscript{\Sigma}$ required for our approach to surpass baseline performance increases as the channel gain difference $z\bob - z\eve$ becomes larger.
\begin{figure}[!htpb]
	\centering
	\includegraphics[width=\linewidth]{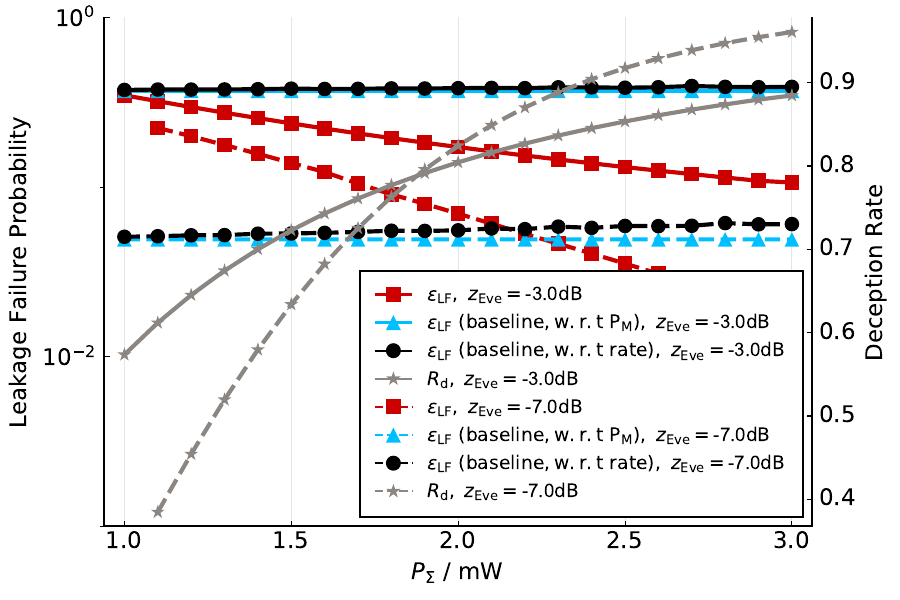}
	\caption{Benchmark results}
	\label{fig: benchmark results}
\end{figure}

\revise{}{Generally, the performance gain of our \ac{pld} scheme regarding conventional \ac{pls} baselines is attributed to its extra degree of freedom, which is introduced by the decomposition of one confidential message into two independent parts, i.e., the ciphertext and the key. With the partially decoupled dependencies of $\varepsilon\subscript{Bob,M}$, $\varepsilon\subscript{Bob,K}$, $\varepsilon\subscript{Eve,M}$, and $\varepsilon\subscript{Eve,K}$ regarding the specification of $\left(P\subscript{M}, P\subscript{K}\right)$, \ac{pld} is enabled with a more flexible resource allocation strategy for enhanced performance.}

In the previous experiments, we maintained the setup $\varepsilon\superscript{th}\lf=0.5$. Given that the threshold of \ac{lfp} $\varepsilon\superscript{th}\lf$ influences the feasible region and potentially the optimum's value, we designed an experiment to explore the impact of this constraint on the deception rate. Specifically, we focused on how the deception rate $R\subscript{d}$ changes with respect to $\varepsilon\superscript{th}\lf$. 

We set $P_{\Sigma}=\SI{2}{\milli\watt}$ and measured $\varepsilon \lf$ as well as $R\subscript{d}$ regarding $\varepsilon\superscript{th}\lf$ under various $z\eve$. Remarkably, neither $R\subscript{d}$ nor $\varepsilon\subscript{LF}$ is influenced by $\varepsilon\subscript{LF}\superscript{th}$. Their optima remain constants under certain channel conditions, as listed in \revise{Tab.}{Table}~\ref{tab: benchmark_feasible_lfp_th results}. Nevertheless, it is worth noting that the selection of $\varepsilon\subscript{LF}\superscript{th}$ significantly impacts the feasible region size of the problem. Given a certain transmission power budget, the feasible region shrinks with decreasing $\varepsilon\superscript{th}\lf$. In fact, when $z\eve = -\SI{3}{\dB}$ and $\varepsilon\superscript{th}\lf=0.1$, no feasible region exists under the constraints (\ref{con:err_bob_message_threshold}--\ref{con:LFP_threshold}). In such cases, one option is to accept a sub-optimal solution with reduced transmission power, where $P\subscript{M} + P\subscript{K} < P\subscript{\Sigma}$. Alternatively, one can adjust the blocklength $n$ of each packet.

\begin{table}[!htpb]
    \centering
	\caption{Benchmark results regarding $\varepsilon\lf\superscript{th}$ with $P_{\Sigma}=\SI{2}{\milli\watt}$}
	\label{tab: benchmark_feasible_lfp_th results}
	\begin{tabular}{m{3cm}m{1.3cm}m{1.3cm}m{1.3cm}}
		\toprule[2px]
		$z\eve$     & $\SI{-3}{\dB}$ &  $\SI{-5}{\dB}$  &  $\SI{-7}{\dB}$ \\         \midrule[1px]
		$\varepsilon\lf$ & $0.1886^{*}$                & 0.0964                            & 0.1003                                       \\ 
		$\varepsilon\lf$ (baseline w.r.t $P\subscript{M}$) &    0.3708               & 0.1611                            & 0.0492                                     \\ 
		$\varepsilon\lf$   (baseline w.r.t rate) & 0.3840   &      0.1732            & 0.0557                 \\
		$R\subscript{d}$  & $0.7989^{*}$  & 0.8800  & 0.8163 \\ 
        \midrule[1px]
        \multicolumn{4}{m{8cm}}{*: The feasible region vanishes with full-power transmission, sub-optimum taken instead.}\\
		\bottomrule[2px]
	\end{tabular}
\end{table}

Additionally, we designed experiments to investigate how the minimum $\varepsilon\superscript{th}\lf$, required to ensure the existence of a feasible solution, varies with the total power $P\subscript{\Sigma}$. The results are depicted in Fig. \ref{fig: benchmark_feasible_lfp_th results}. From the figure, it is evident that as the total power $P\subscript{\Sigma}$ increases, the minimum $\varepsilon\superscript{th}\lf$ required to maintain a feasible region becomes progressively smaller.

\begin{figure}[!htpb]
	\centering
	\includegraphics[width=\linewidth]{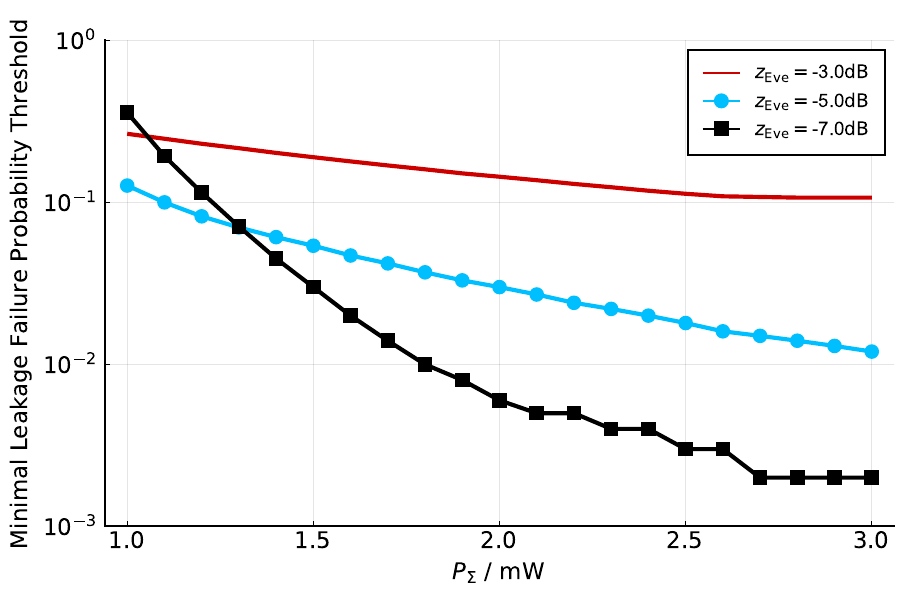}
	\caption{The minimal \ac{lfp} threshold to have a feasible region.}
	\label{fig: benchmark_feasible_lfp_th results}
\end{figure}

\section{Discussions}\label{sec:discussion}
\subsection{Use Cases}\label{subsec:use_cases}
Upon the specific use scenario, our \ac{pld} approach can be applied on either the \ac{up} or the \ac{cp}. For \ac{up} application scenarios, merely a single counterfeit message shall suffice to significantly undermine the eavesdropper's interests, relying solely on the effectiveness of deception.  Typical use cases of this kind are including, but not limited to, military communications, police operations, and confidential business negotiations. When applied on the \ac{cp}, in contrast, the focus is inducing the eavesdropper to expose itself, which relies on an appropriate radio interface protocol design that well merges the \ac{pld} strategy with its authentication procedure.

\subsection{Ciphering Codec Design}\label{subsec:codec_design}
\revise{For the deceptive ciphering algorithm $f:\mathbb{P}\times\mathbb{K}\to\mathbb{M}$, it is essential to }{Intrinsically being a specialized \ac{pls} approach, \ac{pld} does not aim to replace traditional cryptographic methods but to complement them. Generally, our proposed \ac{pld} scheme can be flexibly combined with various traditional cryptographic methods to provide a multi-layer security protection. Nevertheless, though \ac{pld} is not dedicated with any specific cryptographic scheme, it does require its ciphering algorithm deceptive ciphering algorithm $f:\mathbb{P}\times\mathbb{K}\to\mathbb{M}$ to essentially} satisfy Eq.~\eqref{eq:undetectable_ciphering} and Eq.~\eqref{eq:unreplacable_key}. The former ensures that the eavesdropper cannot estimate the cipherer activation status from the decoded ciphertext $m$, and the latter invalidates the eavesdropping strategy of attempting to decrypt the ciphertext with a random key. However, these essential requirements propose a challenge to the ciphering codec design, especially when \ac{pld} is applied on the \ac{up}. On the one hand, it can become a conundrum to satisfy both of them when the cardinality of $\mathbb{M}$ is large. On the other hand, a small $\left\vert\mathbb{M}\right\vert$ will harshly limit the amount of information carried by each single message. Though this may not be a serious issue for the on \ac{cp} where only limited amount of commands are available, it will be a significant challenge for generic \ac{up} application scenarios where a large amount of information needs to be transmitted. Forcing to use a $\mathbb{M}$ with small cardinality in such scenarios will break semantically complete information into multiple codewords, which not only reduces the impairment of the eavesdropper's interests that can be caused by one single false message, but also allows the eavesdropper to leverage its semantic knowledge for coherence analysis on the multiple messages, and therewith avoid being deceived.

A potential solution to this challenge is to combine \ac{pld} with semantic communications. By deploying paired semantic encoder and decoder on both sides of the communication link, it is not only significantly reducing the raw data rate required to deliver the same amount of semantic information, but also effectively constraining the feasible region of codewords.

\revise{}{\subsection{Imperfect Channel State Information}\label{subsec:imperfect_CSI}
Earlier in Sec.~\ref{subsec:ch_dynamics}, we have discussed the adaptation of \ac{pld} scheme to channel dynamics based on real-time \ac{csi} update. Nevertheless, it shall be remarked that even a periodic channel measurement cannot ensure the perfect \ac{csit} in practical scenarios.

While \emph{Bob}'s \ac{csi} is often easy to measure, \emph{Eve}'s \ac{csi} is commonly unobservable by \emph{Alice}. A common approach to address this issue is to consider \emph{Alice} possessing the statistical \ac{csi} of \emph{Eve}, which can be estimated from context information such like the radio environment and the user mobility model. The \ac{pld} scheme can be then optimized regarding the average performance based on such statistical knowledge.

Moreover, limited by the \ac{csi} updating rate, even \emph{Bob}'s \ac{csi} is not always up-to-date, especially when the channel dynamics are high. An outdated \ac{csit} will certainly lead to suboptimal configuration of the \ac{pld} scheme, which may result in performance degradation. Leveraging the insights into \ac{fbl} systems with imperfect \ac{csit} from literature~\cite{HSG2018optimal,KMD+2024SWIPT}, a dedicated performance analysis of \ac{pld} under outdated \ac{csi} can be interesting for future research. As a potential pillar, the emerging techniques of \ac{ai} based fading channel prediction can be promising in compensating the performance loss due to outdated \ac{csit}~\cite{JS2020deep}.}


\subsection{Peak-to-Average Power Ratio}
Regarding Theorem \ref{theorem:max_power} that our approach always prefers full-power transmission, it is not only beneficial for the simplification of the optimization, but also for the reduction of the \ac{papr} of the transmitted signal. Regardless the measured legitimate channel gain $z\subscript{Bob}$ or the estimated eavesdropping channel gain $z\subscript{Eve}$, the transmitter always maintains a consistent transmission power across all messages. Combined with symbol-level \ac{papr} reducing techniques, such as \ac{dft-s-ofdm}, it can provide an outstanding performance in terms of power efficiency and linearity of the power amplifier, which is crucial in practical implementation of wireless transceivers.

\subsection{Orthogonal Frequency-Division Multiplexing}\label{subsec:ofdm}
While \ac{nom} is promising in terms of performance, it lacks compatibility with conventional wireless standards. Adopting our \ac{pld} approach in \ac{ofdm} systems appear therefore an attractive alternative. In an \ac{ofdm} frame, the radio resource, managed in terms of \acp{prb}, can be allocated between the ciphertext and the key in the time-frequency domain. This design frees the receiver from the \ac{sic} operation, as the key and the ciphertext are independently decoded in parallel, which reduces the hardware complexity. However, unlike the transmission power that can be arbitrarily divided, the \acp{prb} can only be allocated in integer numbers, which may lead to a less optimality in comparison to the \ac{nom} solution.

\subsection{Multi-Access}\label{subsec:multi-access}
Though this work mainly focuses on the point-to-point communication scenario, the \ac{pld} approach can also be extended to wireless network scenarios, where the multi-access scheme must be well considered.

The multi-access solution is strongly related to the selection of ciphertext-key multiplexing scheme. As cascaded \ac{sic} is likely leading to a high error rate in key decoding, we do not recommend applying \ac{noma} on top of \ac{nom}-based \ac{pld}, but \ac{oma} solutions such like \ac{ofdma} or simple \ac{tdma}. However, if the ciphertext and key are orthogonally multiplexed like discussed in Sec.~\ref{subsec:ofdm}, \ac{noma} can be considered as a feasible solution to achieve a higher spectral efficiency. \revise{}{However, this also complicates the optimization approach for allocating the power and choosing the key size.} 

\revise{}{Regarding the strategy optimization, a straightforward extension of the single-link optimization problem \reviseRT{$(OP)$}{Problem~\eqref{op:origin}} may be computationally infeasible due to the linear growth of the degree of freedom w.r.t. the number of users. A potential solution is to decompose the optimization task into two stages, where the first stage executes on the \ac{mac} layer to allocate radio resources among users, and the second stage performs on the \ac{phy} layer to optimize the power and key size for each user as resolved in this work. This two-stage optimization approach allows our proposed link-level solution to run on top of various \ac{mac} layer resource allocation strategies, which can be appropriately selected upon the system use scenario, and significantly reduces the computational complexity and improve the scalability of the \ac{pld} solution.}

\section{Conclusion}\label{sec:conclusion}

In this work, we have proposed a comprehensive design for a novel \ac{pld} approach that integrates Physical Layer Security \ac{pls} with deception technology.
 Jointly optimizing the transmission power and encryption key length, we are able to maximize the effective deception rate under a given constraint of \ac{lfp}, simultaneously achieving both secrecy and reliability of communication. We have proved that the optimal power allocation always fully utilizes the transmission power budget, and proposed an efficient algorithm to solve the corresponding optimization problem. The numerical results have demonstrated the superiority of our approach over conventional \ac{pls} solutions in terms of both secrecy and deception performance. Further, we have discussed the potential use cases, the challenges in ciphering codec design, the benefits of the full-power transmission scheme, and the potential extensions of our approach to \ac{ofdm} systems and multi-access scenarios.

\section*{Acknowledgment}
This work is supported in part by the German Federal Ministry of Education and Research in the programme of "Souverän. Digital. Vernetzt." joint projects 6G-ANNA (16KISK105/16KISK097), 6G-RIC (16KISK030/16KISK028) and Open6GHub (16KISK003K/16KISK004/16KISK012),  and in part by the European Commission via the Horizon Europe project Hexa-X-II (101095759). 




\appendices
\section{Proof of Theorem~\ref{theorem:max_power}}
\label{proof:max_power}
	\begin{proof}
		This theorem can be proven by the contradiction. First, with a given $d\subscript{K}$, we define the following auxiliary function to ease the notation:
		\begin{equation}
			\varepsilon_{i,j}(P)\triangleq\varepsilon_{i,j}\vert_{P_j=P}, \forall (i,j)\in\{\text{Bob},\text{Eve}\}\times\{\text{M},\text{K}\},
		\end{equation}
		Suppose there exists an optimal power allocation $\left(P\superscript{o}\subscript{K}, P\superscript{o}\subscript{M}\right)$ that leaves from the power budget a positive residual $P_\Delta=P_\Sigma-P\superscript{o}\subscript{K}-P\superscript{o}\subscript{M}>0$. Since it is optimal, for any feasible power allocation $\left(P\subscript{M}, P\subscript{K}\right)$ it must hold that
		\begin{equation}
			R\subscript{d}\left(P\superscript{o}\subscript{M}, P\superscript{o}\subscript{K}\right)\geqslant R\subscript{d}\left(P\subscript{M}, P\subscript{K}\right).
		\end{equation}
		Meanwhile, there is always another feasible allocation $(P\superscript{f}\subscript{M}, P\superscript{o}\subscript{K})$ where $P\superscript{f}\subscript{M}=P\superscript{o}\subscript{M}+P_\Delta$. Given the same $P\subscript{K}$, it can be straightforwardly shown that $\varepsilon\subscript{Bob,M}$ and $\varepsilon\subscript{Eve,M}$ are monotonically decreasing 
		in $P\subscript{M}$ with:
  \begin{equation}
  \begin{split}
  \label{eq:dR_d_dP_M}
      \frac{\partial R\subscript{d}}{\partial P\subscript{M}}
      &=-\frac{\partial \varepsilon\subscript{Bob,M}}{\partial P\subscript{M}}(1-\varepsilon\subscript{Eve,M})\varepsilon\subscript{Eve,K}\\
      &~~~~~~~~-\frac{\partial \varepsilon\subscript{Eve,M}}{\partial P\subscript{M}}
      (1-(1-\varepsilon\subscript{Bob,M}))\varepsilon\subscript{Bob,K}\\
      &\geqslant 0.
  \end{split}
  \end{equation}
The inequality holds since we have the following derivative based on the chain rule:
\begin{equation}\label{eq:decreasing_errors}
    \begin{split}
        \frac{\partial \varepsilon_{i,\mathrm{M}}}{\partial P\subscript{M}}&=
        \frac{\partial \varepsilon_{i,\mathrm{M}}}{\partial \omega_{i,\mathrm{M}}}
        \frac{\partial \omega_{i,\mathrm{M}}}{\partial \gamma_{i,\mathrm{M}}}
        \frac{\partial \gamma_{i,\mathrm{M}}}{\partial P\subscript{M}}
        \leqslant 0,
    \end{split}
\end{equation}
where $\omega_{i,j}=\sqrt{\frac{n}{V(\gamma_{i,j})}}(\mathcal{C}(\gamma_{i,j})-\frac{d_j}{n})$ is a auxiliary function. Therefore, it always holds that 
		\begin{align}
			\varepsilon\subscript{Bob,M}(P\superscript{o}\subscript{M})&>\varepsilon\subscript{Bob,M}\left(P\superscript{f}\subscript{M}\right),\\
			\varepsilon\subscript{Eve,M}(P\superscript{o}\subscript{M})&>\varepsilon\subscript{Eve,M}\left(P\superscript{f}\subscript{M}\right)
		\end{align}
	 	Then, we have:
		\begin{equation}
			\begin{split}	&R\subscript{d}\left(P\superscript{f}\subscript{M}, P\superscript{o}\subscript{K}\right)-R\subscript{d}\left(P\superscript{o}\subscript{M}, P\superscript{o}\subscript{K}\right)\\
				=&\left[1-2\varepsilon\subscript{Bob,K}\left(P\superscript{o}\subscript{K}\right)\right]\left[\varepsilon\subscript{Bob,M}\left(P\superscript{o}\subscript{M}\right)-\varepsilon\subscript{Bob,M}\left(P\superscript{f}\subscript{M}\right)\right]\\
				&-\left[1-2\varepsilon\subscript{Eve,K}\left(P\superscript{o}\subscript{K}\right)\right]\left[\varepsilon\subscript{Eve,M}\left(P\superscript{o}\subscript{M}\right)-\varepsilon\subscript{Eve,M}\left(P\superscript{f}\subscript{M}\right)\right]\\
				>&0.
			\end{split}
		\end{equation}
		The inequality above holds, since $2\varepsilon\subscript{Bob,K}\leqslant 2\varepsilon\superscript{th}\subscript{Bob,K}< 1$ and $2\varepsilon\subscript{Eve,K}\geqslant 2\varepsilon\superscript{th}\subscript{Eve,K}> 1$. In other words, the solution $P\superscript{f}\subscript{K}$ and $P\superscript{f}\subscript{M}$ achieves a better deception rate  $R\subscript{d}\left(P\superscript{f}\subscript{M},P\superscript{o}\subscript{K}\right)$ than $R\subscript{d}\left(P\superscript{o}\subscript{M},P\superscript{o}\subscript{K}\right)$, which violates the assumption of optimum.
	\end{proof}
\section{Proof of Lemma~\ref{lemma:convexity}}
\label{proof:approximation}
\begin{proof}
    First, we introduce a constant $\lambda(\hat{d}\subscript{K}^{(q)},\hat{P}\subscript{M}^{(q)})=\frac{(1-\varepsilon\subscript{Bob,D}(\hat{d}\subscript{K}^{(q)},\hat{P}\subscript{M}^{(q)}))}{\varepsilon\subscript{Eve,D}((\hat{d}\subscript{K}^{(q)},\hat{P}\subscript{M}^{(q)}))}$ at the local point $(\hat{d}\subscript{K}^{(q)},\hat{P}\subscript{M}^{(q)})$. Since $0\leqslant \varepsilon_{i,\mathrm{D}} \leqslant 1$, it is trivial to show that $\lambda$ is always non-negative. Then, we have:
    \begin{equation}
        R\subscript{d}=\frac{\lambda}{\lambda}(1-\varepsilon\subscript{Bob,D})\varepsilon\subscript{Eve,D}=\frac{1}{\lambda} (1-\varepsilon\subscript{Bob,D})\cdot \lambda\varepsilon\subscript{Eve,D}
    \end{equation}
    Then, based on the inequality of arithmetic and geometric means, we can reconstruct the upper-bound of $R\subscript{d}$ as:
    \begin{equation}
        \begin{split}
           &\sqrt{\lambda R\subscript{d}}\leqslant \frac{(1-\varepsilon\subscript{Bob,D})+\lambda\varepsilon\subscript{Eve,D}}{2}\\
           \Longleftrightarrow &\lambda R\subscript{d}\leqslant 
            \left(
                \frac{(1-\varepsilon\subscript{Bob,D})+\lambda\varepsilon\subscript{Eve,D}}{2}
            \right)^2\\
            \Longleftrightarrow & R\subscript{d}\leqslant\frac{1}{4 \lambda }            \left(
                (1-\varepsilon\subscript{Bob,D})+\lambda\varepsilon\subscript{Eve,D}
            \right)^2\triangleq \hat{R}\subscript{d},
        \end{split}
    \end{equation}
    which completes the proof. 
\end{proof}

\section{Proof of Lemma~\ref{lemma:convex_d_K}}

\begin{proof}
\label{proof:convex_d_K}
We start with  the objective function $\hat{R}_d$. 
To prove its convexity, we first investigate the monotonicity of  $\varepsilon_{i,j}$ with respect to $d\subscript{K}$. In particular, we have
    \begin{equation}
        \frac{\partial \varepsilon_{i,\mathrm{M}}
    }{\partial d\subscript{K}}=0,
\end{equation}
and
    \begin{equation}
    \begin{split}
            \frac{\partial \varepsilon_{i,\mathrm{K}}
        }{\partial d\subscript{K}}&=\frac{\partial \varepsilon_{i,\mathrm{K}}}{\partial \omega_{i,\mathrm{K}}}\frac{\partial \omega_{i,\mathrm{K}}}{\partial d\subscript{K}}\\
        &=e^{-\frac{\omega^2_{i,K}}{2}}\sqrt{\frac{n}{V_{i,K}}}\cdot\frac{1}{n}\geqslant 0.
    \end{split}
\end{equation}
Therefore, $\varepsilon\subscript{Bob,K}$ and $\varepsilon\subscript{Eve,K}$ are monotonically increasing in $d\subscript{K}$.
Then, we further investigate the convexity of $\varepsilon_{i,j}$. Their second derivatives are: 
            \begin{equation}
        \frac{\partial^2 \varepsilon_{i,\mathrm{M}}
    }{\partial^2 d\subscript{K}}=0,
    \end{equation}
    and
    \begin{equation}
    \begin{split}
        \frac{\partial^2 \varepsilon_{i,\mathrm{K}}
    }{\partial^2 d\subscript{K}}
    &=\frac{\partial^2 \varepsilon_{i,\mathrm{K}}}{\partial \omega_{i,\mathrm{K}}^2}
    \underbrace{\left(
        \frac{\partial \omega_{i,\mathrm{K}}}{\partial d\subscript{K}}
    \right)^2}_{\geqslant 0}
    +\underbrace{\frac{\partial \varepsilon_{i,\mathrm{K}}}{\partial \omega_{i,\mathrm{K}}}
        \frac{\partial^2 \omega_{i,\mathrm{K}}}{\partial d\subscript{K}^2}}_{=0}.
    \end{split}
\end{equation}
Note that $\varepsilon_{i,j}(\omega_{i,j})$ is a Q-function, which is convex if $\omega_{i,j}$ is non-negative and concave if it is non-positive. 
It indicates that $\varepsilon_{i,\mathrm{K}}$ is convex in $d\subscript{K}$ if $\omega_{i,\mathrm{K}}\geqslant 0$ while being concave if $\omega_{i,\mathrm{K}}\leqslant 0$. 
Recall that the transmission must fulfill $\varepsilon\subscript{Bob,K}\leqslant \varepsilon\subscript{Bob,K}\superscript{th}\leqslant 0.5$ and $\varepsilon\subscript{Eve,K}\geqslant \varepsilon\subscript{Eve,K}\superscript{th} \geqslant 0.5$. Therefore, for any feasible $d\subscript{K}$, it must hold that 
\begin{equation}
    \omega\subscript{Bob,K}=\sqrt{\frac{n}{V(\gamma\subscript{Bob,K}}}(\mathcal{C}(\gamma\subscript{Bob,K}-\frac{d\subscript{K}}{n})\geqslant Q^{-1}(\varepsilon\subscript{Bob,K}\superscript{th})\geqslant 0,
\end{equation}
and
\begin{figure*}[!htbp]
    \setcounter{equation}{\theequation+3}
	\begin{equation}
\label{eq:d2e_dP2}
		\begin{split}
			&\frac{\partial^2 \varepsilon_{i,\mathrm{M}}}{\partial P\subscript{M}^2}
			=
			\frac{
				\partial \varepsilon_{i,\mathrm{M}}
			}
			{
				\partial \gamma_{i,\mathrm{M}}
			}
			\frac{
				\partial^2 \gamma_{i,\mathrm{M}}
			}
			{
				\partial P\subscript{M}^2
			}
			+\frac{
				\partial^2 \varepsilon_{i,\mathrm{M}}
			}
			{
				\partial \gamma^2_{i,\mathrm{M}}
			}
			\left(
			\frac{
				\partial \gamma_{i,\mathrm{M}}
			}
			{
				\partial P\subscript{M}
			}
			\right)^2\\
			=&\frac{1}{\sqrt{2\pi}}e^{-\frac{\omega_{i,\mathrm{M}}}{2}}
			\left(
			\frac{\partial \omega_{i,\mathrm{M}}}{\partial \gamma_{i,\mathrm{M}}}
			\left(
			\omega_{i,\mathrm{M}}
			\frac{\partial \omega_{i,\mathrm{M}}}{\partial \gamma_{i,\mathrm{M}}}
			\left(
			\frac{\partial \gamma_{i,\mathrm{M}}}{\partial P\subscript{M}}        
			\right)^2
			-\frac{\partial^2 \gamma_{i,\mathrm{M}}}{\partial P\subscript{M}^2}
			\right)
			-\underbrace{
				\frac{\partial^2 \omega_{i,\mathrm{M}}}{\partial \gamma^2_{i,\mathrm{M}}}
				\left(
				\frac{\partial \gamma_{i,\mathrm{M}}}{\partial P\subscript{M}}
				\right)^2
			}_{\leqslant 0}
			\right)\\
			\geqslant&\frac{1}{\sqrt{2\pi}}e^{-\frac{\omega_{i,\mathrm{M}}}{2}}
			\frac{\partial \omega_{i,\mathrm{M}}}{\partial \gamma_{i,\mathrm{M}}}
			\left(
			\omega_{i,\mathrm{M}}
			\frac{\partial \omega_{i,\mathrm{M}}}{\partial \gamma_{i,\mathrm{M}}}
			\left(
			\frac{\partial \gamma_{i,\mathrm{M}}}{\partial P\subscript{M}}        
			\right)^2
			-\frac{\partial^2 \gamma_{i,\mathrm{M}}}{\partial P\subscript{M}^2}
			\right)\\
			=&\frac{1}{\sqrt{2\pi}}e^{-\frac{\omega_{i,\mathrm{M}}}{2}}
			\frac{\partial \omega_{i,\mathrm{M}}}{\partial \gamma_{i,\mathrm{M}}}
			\left(
			\omega_{i,\mathrm{M}}
			\frac{\partial \omega_{i,\mathrm{M}}}{\partial \gamma_{i,\mathrm{M}}}
			\frac{
				z^2_1(z_iP_\Sigma+\sigma^2_i)^2
			}{
				(z_iP\subscript{K}+\sigma^2_i)^4
			}
			-\frac{
				2z_i^2(z_iP_\Sigma+\sigma^2_i)
			}{
				(z_iP\subscript{K}+\sigma^2_i)^3
			}
			\right)\\
			=&\frac{1}{\sqrt{2\pi}}e^{-\frac{\omega_{i,\mathrm{M}}}{2}}
			\frac{\partial \omega_{i,\mathrm{M}}}{\partial \gamma_{i,\mathrm{M}}}
			\frac{z_i^2(z_iP_\Sigma+\sigma^2_i)}{(z_iP\subscript{K}+\sigma^2_i)^3}
			\left(
			\omega_{i,\mathrm{M}}
			\frac{\partial \omega_{i,\mathrm{M}}}{\partial \gamma_{i,\mathrm{M}}}
			\underbrace{
				\frac{
					(z_iP_\Sigma+\sigma^2_i)
				}{
					(z_iP\subscript{K}+\sigma^2_i)
				}
			}_{\geqslant \gamma_{i,\mathrm{M}}}
			-2
			\right)\\
			\geqslant&\frac{1}{\sqrt{2\pi}}e^{-\frac{\omega_{i,\mathrm{M}}}{2}}
			\frac{\partial \omega_{i,\mathrm{M}}}{\partial \gamma_{i,\mathrm{M}}}
			\frac{z_i^2(z_iP_\Sigma+\sigma^2_i)}{(z_iP\subscript{K}+\sigma^2_i)^3}
			\left(
			\omega_{i,\mathrm{M}}
			\frac{\partial \omega_{i,\mathrm{M}}}{\partial \gamma_{i,\mathrm{M}}}
			\gamma_{i,\mathrm{M}}-2
			\right)\geqslant\frac{6.25\gamma_{i,\mathrm{M}}}{\sqrt{\gamma_{i,\mathrm{M}}(\gamma_{i,\mathrm{M}}+1)}}-2\geqslant 0,
		\end{split}
	\end{equation}%
\hrule%
\end{figure*}
\begin{equation}
	\setcounter{equation}{\theequation-4}
    \omega\subscript{Eve,K}=\sqrt{\frac{n}{V(\gamma\subscript{Eve,K}}}(\mathcal{C}(\gamma\subscript{Eve,K}-\frac{d\subscript{K}}{n})\leqslant Q^{-1}(\varepsilon\subscript{Eve,K}\superscript{th})\leqslant 0,
\end{equation}
where $Q^{-1}(\cdot)$ is the inverse Q-function. Therefore, $\varepsilon\subscript{Bob,K}$ is a convex and decreasing function while $\varepsilon\subscript{Eve,K}$ being a concave and decreasing function.  
Recall that $\hat{R}_d$  is the quadratic function of $\frac{1}{R_b}$, $\frac{1}{1-\varepsilon\subscript{Eve,K}}$ and $\frac{1}{\varepsilon\subscript{Eve,M}}$ according to~\eqref{eq:R_app}. Then, $\hat{R}_d$ is convex, if each of the components, i.e., $R_b$, $(1-\varepsilon\subscript{Eve,M})$ and $\varepsilon\subscript{Eve,K}$, is concave and non-negative. Clearly, both this is true for $(1-\varepsilon\subscript{Eve,M})$ and $\varepsilon\subscript{Eve,K}$. Therefore, we focus on the concavity of $R_b$ with its second derivative, i.e.,  
\begin{equation}
    \begin{split}
        \frac{\partial^2 R_b}{\partial d\subscript{K}^2}=-(1-\varepsilon\subscript{Bob,M})\underbrace{\frac{\partial^2 \varepsilon\subscript{Bob,K}}{\partial d\subscript{K}^2}}_{\geqslant 0}-\underbrace{\frac{\partial \varepsilon\subscript{Bob,K}}{\partial d\subscript{K}}}_{\geqslant 0}\leqslant 0.
    \end{split}
\end{equation}
    Hence, $R_b$ is indeed concave, i.e., $\hat{R}\subscript{d}$ is convex. It is also trivial to show that all the constraints are either convex or linear, i.e., the feasible set of Problem~\eqref{Problem_SCA_BCD_d_K} is convex. Then, since the objective function to be maximized is concave and its feasible set is convex, Problem~\eqref{Problem_SCA_BCD_d_K} is a convex problem. 
\end{proof}

\section{Proof of Theorem~\ref{theorem:convex_P_M}}
\label{proof:convex_P_M}
\begin{proof}

Similar to the proof of Lemma~\ref{lemma:convex_d_K}, we start with the convexity of the objective function $\hat{R}_d$. We also first investigate the monotonicity of $\varepsilon_{i,M}$ with respect $P\subscript{M}$ as follows:
    \begin{equation}
    \label{eq:de_dP}
			\begin{split}
				\frac{\partial \varepsilon_{i,\mathrm{M}}}{\partial P\subscript{M}}
				&=\frac{\partial \varepsilon_{i,\mathrm{M}}}{\partial \gamma_{i,\mathrm{M}}}
                \frac{\partial \gamma_{i,\mathrm{M}}}{\partial P\subscript{M}}
                =\underbrace{\frac{\partial \varepsilon_{i,\mathrm{M}}}{\partial \gamma_{i,\mathrm{M}}}}_{\leqslant 0}\frac{2z_i^2(z_iP_\Sigma+\sigma^2)}{(z_iP\subscript{K}+\sigma^2)^3}\leqslant 0.
			\end{split}
    \end{equation}
    \setcounter{equation}{\theequation+1}
Note that the above inequality holds with the constraints~\eqref{con:err_bob_message_threshold} and~\eqref{con:err_eve_message_threshold}, i.e., $\omega_{i,M}\geqslant 0$. Moreover, since $n\geqslant 10$ and $\varepsilon_{i,\mathrm{M}}\leqslant \varepsilon\superscript{th}_{i,\mathrm{M}}<0.5$, we can proven that $\varepsilon_{i,M}$ is convex in $P\subscript{M}$ with ~\eqref{eq:de_dP} and ~\eqref{eq:d2e_dP2}.  Note that the inequality in~\eqref{eq:d2e_dP2} holds with $\gamma_{i,M}\geqslant \gamma_{th}\geqslant 1$, which is required to fulfill the error probability constraints in practical scenarios~\cite{ZHY+2023joint}.  

Similarly, we can show that $\varepsilon\subscript{Bob,K}$ is convex while $\varepsilon\subscript{Eve,K}$ being concave with\vspace{-5pt} 
\begin{equation}
    \frac{\partial \gamma_{i,\mathrm{K}}}{\partial P\subscript{M}}=-\frac{z_{i}}{\sigma^2},~
    \frac{\partial^{2}\gamma_{i,\mathrm{K}}}{\partial P\subscript{M}^{2}}=0.
    \vspace{-5pt}
\end{equation}

To avoid repetition, we omit the details. Note that the convexity/concavity of $\varepsilon\subscript{Bob,K}$ differs from the ones of $\varepsilon\subscript{Eve,K}$. This is due to the fact that $\omega\subscript{Bob,K}\geqslant 0$ and $\omega\subscript{Eve,K}\leqslant 0$ according to the constraints~\eqref{con:err_bob_key_threshold} and~\eqref{con:err_eve_key_threshold}. 

With the above results, it is clear that $\frac{1}{\varepsilon\subscript{Eve,K}}$ and $\frac{1}{1-\varepsilon\subscript{Eve,M}}$ are both convex in $P\subscript{M}$, since a convex and decreasing function composed with a concave function is also convex~\cite{boyd2004convex}. However, we still need to determine the convexity of $R_b$. In particular, its second-order  derivative is:
\begin{equation}
\begin{split}
\frac{\partial^2R_{b}}{\partial P\subscript{M}^2}  
=&-\frac{\partial^2\varepsilon\subscript{Bob,K}}{\partial P\subscript{M}^2}+\frac{\partial^2\varepsilon\subscript{Bob,M}}{\partial P\subscript{M}^2}\varepsilon\subscript{Bob,K}+\frac{\partial^2\varepsilon\subscript{Bob,K}}{\partial P\subscript{M}^2}\varepsilon\subscript{Bob,M}\\&+\frac{\partial\varepsilon\subscript{Bob,K}}{\partial P\subscript{M}}\frac{\varepsilon\subscript{Bob,M}}{\partial P\subscript{M}}\\  
 =&-(1-\varepsilon\subscript{Bob,M})\underbrace{\frac{\partial^2\varepsilon\subscript{Bob,K}}{\partial P\subscript{M}^2}}_{\geqslant 0}+\underbrace{\frac{\partial^2\varepsilon\subscript{Bob,M}}{\partial P\subscript{M}^2}}_{\geqslant 0}\varepsilon\subscript{Bob,K}\\
 &+\underbrace{\frac{\partial\varepsilon\subscript{Bob,K}}{\partial P\subscript{M}}}_{\geqslant 0}
 \underbrace{\frac{\varepsilon\subscript{Bob,M}}{\partial P\subscript{M}}}_{\leqslant 0}
 \leqslant \frac{\partial^2\varepsilon\subscript{Bob,M}}{\partial P\subscript{M}^2} -\frac{\partial^2\varepsilon\subscript{Bob,K}}{\partial P\subscript{M}^2}.
\end{split}
\end{equation}

Since the feasible $P\subscript{M}$ must fulfill that $\varepsilon_{\text{Bob},j}\leqslant \varepsilon\superscript{th}_{\text{Bob},j}\leqslant 0.1$ and $P\subscript{K}=P_{\sum}-P\subscript{M}$, we have $\varepsilon\subscript{Bob,K}\leqslant 10(1-\varepsilon\subscript{Bob,M})$, $P\subscript{M}\geqslant P\subscript{K}$ and $\gamma\subscript{Bob,M}\leqslant \gamma\subscript{Bob,K}$. Therefore, it also indicates that $\omega(\gamma\subscript{Bob,M},d\subscript{M})\leqslant \omega(\gamma\subscript{Bob,M},d\subscript{K})\leqslant \omega(\gamma\subscript{Bob,K},d\subscript{K})$. Note that both $\varepsilon\subscript{Bob,M}$ and $\varepsilon\subscript{Bob,K}$ are error probability, which is characterized by the Q-function according to~\eqref{eq:FBL_error}. Then, let denote $P_{0.5}$ the transmit power that achieves $\varepsilon\subscript{Bob,K}(P_{0.5})=0.5$ while $P_{\infty}$ the transmit power that achieves $\varepsilon\subscript{Bob,K}(P_{\infty})=0$. It holds that  
\begin{equation}
    \frac{\partial^2 \varepsilon\subscript{Bob,M}}{\partial P_{0.5}^2}\leqslant \frac{\partial^2 \varepsilon\subscript{Bob,K}}{\partial P_{0.5}^2} = 0,~
    \frac{\partial^2 \varepsilon\subscript{Bob,M}}{\partial P_{\infty}^2} = \frac{\partial^2 \varepsilon\subscript{Bob,K}}{\partial P_{\infty}^2} > 0,
\end{equation}

Therefore, within the feasible set of Problem~\eqref{Problem_SCA_BCD_P_M}, we have $\frac{\partial^2 R_b}{\partial P\subscript{M}^2}\leqslant \frac{\partial^2\varepsilon\subscript{Bob,M}}{\partial P\subscript{M}^2} -\frac{\partial^2\varepsilon\subscript{Bob,K}}{\partial P\subscript{M}^2}\leqslant 0$.

Since all components of $\hat{R}_d$ are convex, the objective function $\hat{R}_d$ is also convex. Moreover, it is trivial to show that all constraints are either convex or affine. As a result, Problem~\eqref{Problem_SCA_BCD_P_M} is convex.  

\end{proof}

\end{document}